\documentclass[journal]{IEEEtran}
\IEEEoverridecommandlockouts
\usepackage[utf8]{inputenc}
\usepackage{colortbl} 
\usepackage[T1]{fontenc}
\usepackage{textcomp}
\usepackage{fixltx2e}
\usepackage{microtype}
\usepackage{calc}
\usepackage[normalem]{ulem}
\usepackage{makecell}
\usepackage{diagbox}
\usepackage{pifont}
\usepackage{gensymb}
\usepackage{comment}
\usepackage{bbm}
\usepackage{amsthm}
\newtheorem{Theorem}{Theorem}

\usepackage{amsmath,amssymb,amsfonts}
\usepackage{algorithm}
\usepackage{algorithmic}
\usepackage{multirow}
\usepackage{textcomp}
\usepackage{xcolor}\usepackage[dvipsnames]{xcolor}
\usepackage{url}
\usepackage{ tipa }
\usepackage[range-phrase=--,per-mode=symbol-or-fraction,binary-units=true,range-units=single,list-units=single,detect-all]{siunitx}
\usepackage[capitalize,noabbrev]{cleveref}
\crefformat{footnote}{#2\footnotemark[#1]#3}

\usepackage{booktabs}
\usepackage{url}
\RequirePackage{xstring}
\RequirePackage{xparse}
\RequirePackage[]{acro}
\NewDocumentCommand\acrodef{mO{#1}mG{}}{\DeclareAcronym{#1}{short={#2}, long={#3}, #4}}
\usepackage{siunitx}

\DeclareUnicodeCharacter{0301}{\'{e}}
\usepackage{acro}
\DeclareAcronym{LMI}{
  short = LMI,
  long  = linear matrix inequality,
  long-plural-form = linear matrix inequalities
}
\acrodef{mmWave}{millimeter-wave}
\acrodef{ISAC}{integrated sensing and communications}
\acrodef{PMCW}{phase-modulated continuous waveform}
\acrodef{FMCW}{frequency-modulated continuous waveform}
\acrodef{OFDM}{orthogonal frequency-division multiplexing}
\acrodef{MIMO}{multiple-input multiple-output}
\acrodef{V2X}{vehicle-to-everything}
\acrodef{PAPR}{peak-to-average power ratio}
\acrodef{SINR}{signal-to-interference-plus-noise ratio}
\acrodef{RDM}{range-Doppler map}
\acrodef{SISO}{single-input single-output}
\acrodef{CRLB}{Cramér-Rao lower bound}
\acrodef{ULA}{uniform linear array}
\acrodef{UL}{uplink}
\acrodef{UE}{user equipment}
\acrodef{AoD}{angle of departure}
\acrodef{AoA}{angle of arrival}
\acrodef{AWGN}{additive white Gaussian noise}
\acrodef{RCS}{radar cross section}
\acrodef{RSU}{roadside unit}
\acrodef{LoS}{line-of-sight}
\acrodef{NLoS}{non-line-of-sight}
\acrodef{FFT}{fast Fourier transform}
\acrodef{IFFT}{inverse fast Fourier transform}
\acrodef{PRI}{pulse repetition interval}
\acrodef{i.i.d.}{independent and identically distributed}
\acrodef{MU-MIMO}{multi-user MIMO}
\acrodef{FIM}{Fisher information matrix}
\acrodef{IIoT}{intelligent internet of things}
\acrodef{CU}{central unit}
\acrodef{MUI}{multi-user interference}
\acrodef{LTE}{long-term evolution}
\acrodef{NR}{new radio}
\acrodef{C-RAN}{cloud radio access network}
\acrodef{RRA}{radio resource allocation}
\acrodef{MSE}{mean squared error}
\acrodef{VA}{virtual array}
\acrodef{EFIM}{equivalent FIM}
\acrodef{TDM}{time-division multiplexing}
\acrodef{RMS}{root mean square}
\acrodef{FDM}{frequency-division multiplexing}
\acrodef{CDM}{code-division multiplexing}
\acrodef{OTFS}{orthogonal time frequency space}
\acrodef{PEB}{position error bound}
\acrodef{VEB}{velocity error bound}
\acrodef{MMS}{multi-monostatic sensing}
\acrodef{MBS}{multi-bistatic sensing}
\acrodef{MXS}{multi-X-static hybrid sensing}
\acrodef{UAV}{unmanned aerial vehicle}
\acrodef{BF}{beamforming}
\acrodef{BS}{base station}
\acrodef{SNR}{signal-to-noise ratio}
\acrodef{eMBB}{enhanced mobile broadband}
\acrodef{MUSIC}{multiple signal classification}
\acrodef{FE}{fine estimation}
\acrodef{Tx}{transmitter}
\acrodef{Rx}{receiver}
\acrodef{TRx}{transceiver}
\acrodef{DL}{downlink}
\acrodef{UP}{uplink}
\acrodef{DoF}{degree of freedom}
\acrodef{ESNR}{energy SNR}
\acrodef{BPSK}{binary phase shift keying}
\acrodef{ICI}{inter-cell interference}
\acrodef{CP}{cyclic prefix}
\acrodef{ISFFT}{inverse symplectic finite Fourier transform}
\acrodef{SFFT}{symplectic finite Fourier transform}
\acrodef{DD}{delay-Doppler}
\acrodef{RU}{radio unit}
\acrodef{SDP}{semidefinite programming}
\acrodef{PSD}{positive semidefinite}
\acrodef{ISD}{inter-side distance}
\acrodef{MRT}{maximum ratio transmission}
\acrodef{NMAE}{normalized mean absolute error}
\acrodef{QoS}{quality of service}
\acrodef{AP}{access point}
\acrodef{CoMP}{coordinated multipoint}
\acrodef{PA}{power allocation}
\acrodef{SE}{spectrum efficiency}
\acrodef{SDR}{semidefinite relaxation}
\acrodef{SOCP}{second-order cone programming}
\acrodef{SOC}{second-order cone}
\acrodef{JT}{joint transmission}
\acrodef{ZF}{zero-forcing}
\acrodef{RZF}{regularized zero-forcing}
\acrodef{NSP}{null-space projection}
\acrodef{BCD}{block coordinate descent}
\acrodef{SCA}{successive convex approximation}
\acrodef{CS}{coordinated scheduling}
\acrodef{CB}{coordinated beamforming}
\acrodef{DPS}{dynamic point selection}
\acrodef{KKT}{Karush–Kuhn–Tucker}
\acrodef{CSI}{channel state information}
\acrodef{PRNS}{pseudo-random noise sequence}
\acrodef{SI}{sensing interference}
\acrodef{ITI}{inter-target interference}
\acrodef{QCQP}{quadratically constrained quadratic program}
\acrodef{IPM}{interior-point method}
\acrodef{NF}{noise figure}
\acrodef{UMi}{urban micro}

\usepackage[caption=false,font=footnotesize]{subfig}
\usepackage[pdftex]{graphicx}
\DeclareGraphicsExtensions{.pdf,.png,.jpg,.tikz}
\usepackage{csquotes}
\usepackage[backend=biber,style=ieee,doi=false,isbn=false,sorting=none,sortcites=true,mincitenames=1,maxcitenames=2]{biblatex}
\DeclareFieldFormat{sentencecase}{#1} 
\DeclareFieldFormat{titlecase}{#1} 
\usepackage{xpatch}
\xpatchbibmacro{textcite}{\addspace}{\addnbspace}{}{}
\xpatchbibmacro{Textcite}{\addspace}{\addnbspace}{}{}
\DefineBibliographyStrings{english}{
	andothers = et~al\adddot\addspace
}

\graphicspath{{./Bilder/}}    
\begin{document}

\title{CRLB-Driven Beamforming and Power Allocation for Multi-BS Cooperative ISAC Networks}

\author{%
Yanpeng Su,~\IEEEmembership{Graduate Student Member,~IEEE}, Maximilian Lübke,~\IEEEmembership{Member,~IEEE},\\Mengyu Zhang,~\IEEEmembership{Graduate Student Member,~IEEE}, and Norman Franchi,~\IEEEmembership{Member,~IEEE},\\
\thanks{This work was supported by BMFTR, Germany, under Open6GHub+ (16KIS2404). The work contributes to the research within the 6G-Valley innovation cluster. \textit{(Corresponding author: Yanpeng Su.)}}
\thanks{Yanpeng Su, Maximilian Lübke, Mengyu Zhang, and Norman Franchi are with the Institute for Smart Electronics and Systems, Friedrich-Alexander-Universität Erlangen-Nürnberg, Erlangen 91058, Germany (email: yanpeng.su@fau.de; maximilian.luebke@fau.de; mengyu.zhang@fau.de; norman.franchi@fau.de).}
}
\markboth{IEEE}{Su \textit{et al.}: CRLB-Driven Beamforming and Power Allocation for Cooperative ISAC Networks}

\maketitle

\begin{abstract}
This paper presents a Cramér–Rao lower bound (CRLB)-driven beamforming (BF) and power allocation (PA) framework for cooperative integrated sensing and communication (ISAC) networks, where a set of multi-antenna base stations (BSs) jointly serve multiple users and simultaneously perform multi-static target estimation. 
In our design, we investigate how the position error bound (PEB) and velocity error bound (VEB) can be exploited and incorporated into BF and PA optimization. The design leveraging PEB and VEB as metrics directly characterizes sensing accuracy.
First, we propose a semidefinite programming (SDP)-based BF, where the PEB and VEB constraints are handled by double Schur complements, and the rank-1 constraints of the BF covariance matrices are relaxed by semidefinite relaxation (SDR), whose tightness is proved using the Karush–Kuhn–Tucker (KKT) conditions. Addressing the complexity, we further develop a two-stage PA algorithm, where the communication and sensing beams are formed by the regularized zero-forcing (RZF) and null-space projection (NSP) methods, respectively, and the communication and sensing PAs are solved sequentially by second-order cone programming (SOCP). Although the PA algorithm sacrifices the degrees of freedom, its performance shows a slight gap compared to BF in the simulation, while the execution time of around 25\,ms demonstrates its applicability in dynamic environments. 

\end{abstract}

\begin{IEEEkeywords}
\text{ }6G, beamforming, cell-free, coordinated multipoint (CoMP), Cramér–Rao lower bound (CRLB), integrated sensing and communications (ISAC), massive multiple-input multiple-output (MIMO), networked ISAC, power allocation
\end{IEEEkeywords}

\section{Introduction}
\Ac{ISAC} is regarded as a key enabling technology for future wireless networks. By allowing communication and sensing functionalities to share the spectrum, hardware, and radio resources, ISAC provides environmental awareness for emerging applications such as intelligent transportation systems, smart home and factory, and intelligent internet of things \cite{10536135,persp}, while avoiding significantly increased spectrum conflict and energy consumption \cite{9737357,10012421}. Unlike conventional coexistence communication and sensing, where the two systems are designed separately, ISAC requires a unified design in signal transmission and processing \cite{10188491}, highlighting the necessity of \ac{RRA} for ISAC-enabled networks.

Multi-\ac{BS} cooperation has become an essential technique since the advent of 4G and 5G, with \ac{CoMP} transmission deployed to enhance the \ac{QoS} at the cell edges and suppress the \ac{ICI} \cite{9272226,3gpp.36.819}. 
With the evolution toward 6G, cooperative ISAC networks are expected to become an important paradigm \cite{10418473}, where multiple nodes jointly serve multiple \acp{UE} and perceive objects. Such an architecture is closely related to the concepts of \ac{CoMP} \cite{10599241}, cell-free \cite{7827017}, and distributed \ac{MIMO} \cite{8835674}, in which the distributed BSs or \acp{AP} are connected to a \ac{CU} for coordinated transmission, reception, and signal processing. Besides the improved throughput and seamless coverage of the communication service, cooperative sensing provides higher sensing spatial diversity and resolution and enables 2D/3D velocity estimation \cite{9585321}.

The cooperation between the BSs or APs in communication can be classified into three types: \ac{JT}, \ac{CS}/\ac{CB}, and \ac{DPS} \cite{9318091}. JT implies that all the BSs/APs are coordinated to transmit data to the same UE simultaneously, which can exploit coherent combining and improve the received \ac{SINR} and data rate. In CS/CB, each UE is served by a single transmission point, while scheduling or \ac{BF} decisions are coordinated across multiple points to mitigate \ac{ICI}.
In \ac{DPS}, the serving transmission point is dynamically selected from a coordinated set, and UE data is transmitted from only one point at a given time. Moreover, in \cite{10599241}, we summarized three types of cooperative sensing: \ac{MMS}, \ac{MBS}, and their hybrid, named \ac{MXS}, which implies that each node can extract sensing information from signals transmitted by itself and other nodes.
\vspace{-1em}

\subsection{Related Works}
To fully exploit the potential of ISAC and improve its efficiency, RRA, especially the \ac{PA} and \ac{BF} (spatial allocation), plays an important role \cite{10233825}. 
Studies primarily choose the \ac{SINR} or \ac{SE} as the communication metric, while studies on sensing metrics have drawn a lot of attention in recent years, where sensing \ac{SNR}/SINR, detection probability, and \ac{CRLB} are widely investigated.
Studies in \cite{9667503,10308585,10742291,10437290,10015103,10000730,10494224,10113889} employ the \ac{SNR} or SINR as the sensing metric. 
The study in \cite{9667503} considers BF in a two \acp{BS} scenario, in which one BS serves as the ISAC transmitter and receiver (TRx), while the other performs passive sensing. 
\Ac{ICI} suppression is addressed in \cite{10308585}, where the communication system operates in a CB mode, and the sensing type of MMS is considered.
In \cite{10742291}, \ac{SDP}-based BF and PA are designed for a cell-free MIMO ISAC system, where the BF is employed as a benchmark for the evaluation of PA performance, while the PA benefits from lower complexity. The system operates in a \ac{JT}-MXS mode, maximizing the channel capacity and information utilization efficiency.
The study in \cite{10437290} uses \ac{SCA} with lower complexity than SDP for BF. The study in \cite{10015103} introduces a joint BF and clustering problem, which is solved by SDP- and SCA-based algorithms, respectively. 
In \cite{10000730,10494224}, the authors designed a PA algorithm for a JT-MBS system, where dedicated APs are deployed for sensing reception. The BF employs classical \ac{RZF} for communication beams and \ac{NSP} for sensing beams.
A deep learning-based method can be found in \cite{10113889}, where each UE and target is associated with a single BS.

Related works with task-oriented sensing performance indicators primarily focus on detection performance or position/velocity estimation accuracy. While target detection probability is usually determined by the overall SINR or \ac{SNR} \cite{10304081,10380513}, estimation accuracy of multistatic sensing cannot be fully characterized by these metrics since the sensing accuracy is also affected by spatial characteristics like node distribution as well as target position and velocity \cite{su2026}. \ac{CRLB} is the most common metric applied in parameter estimation-oriented studies since it expresses the theoretical lower bound of \ac{MSE} of sensing parameters.  
Studies in \cite{10942860,10615966} provide solutions respectively for PA and BF with angle CRLB.
In \cite{10577579}, the authors actually address two scenarios: CB-MMS, where the BF is designed to suppress the ICI; JT-MXS, where the signals are coherently superposed at UEs, while the BSs can obtain sensing information from other BSs' signals. The problem is solved by an SDP-based algorithm, where the angle-based \ac{PEB} constraint is processed by Schur complement. On the contrary, studies in \cite{8835674,10292936} are based on single-antenna networks, where the position information is completely contributed by range information, and the convex optimization problem is formulated by relaxing the PEB constraint \cite{5753953}, leading to a more conservative solution. 
A similar PA problem is addressed in \cite{9842350}, where the \ac{SDR} technique avoids the conservatism but increases the complexity. The studies in \cite{10436719,11031778} address the PEB based on both range and angle information, and the solution is also based on the SDP with Schur complement. Addressing the requirement of velocity estimation, in \cite{xia2026}, a PA algorithm for single-antenna AP networks with the sensing constraints on both PEB and \ac{VEB} is designed, whereas the single-antenna structure abandons angle information.
\vspace{-1em}

\subsection{Contributions}
Motivated by the accuracy requirements for wireless sensing in 3GPP Release 19 \cite{3gpp22137}, this work employs both PEB and VEB as sensing metrics in BF and PA for cooperative ISAC networks. Related works \cite{10942860,10615966,10577579,8835674,10292936,9842350,10436719,11031778,xia2026} exhibit limitations in handling velocity constraints or position information integrity.
As future 6G mobile networks and ISAC systems will inherently adopt massive \ac{MIMO} and \ac{eMBB} technologies \cite{8869705,9737357}, the CRLB of networked ISAC should account for information of \ac{AoA}, delay/range, and Doppler/radial velocity. 
In our recent work \cite{su2026}, we derived full formulas of \ac{PEB} and \ac{VEB} for cooperative MIMO ISAC networks, where the coupling between the position and velocity information is captured. 
Addressing the tractability and complexity, simplified versions of PEB and VEB were developed with small deviation from the full CRLB. Building on these results, this work leverages the derived CRLBs as the sensing metric in BF and PA designs. 

We consider the \ac{DL} transmission, where the \ac{eMBB} and massive MIMO technologies result in much higher spatial resolution than \ac{UL} sensing, and the synchronization and privacy issues are avoided \cite{9585321}. The system is assumed to operate in a JT-MXS mode, where the coherent transmission benefits from improved SE, and the MXS enhances spatial diversity and accuracy to the greatest extent possible. The main contributions are summarized as follows:
\begin{itemize}
    \item \textbf{CRLB for networked ISAC:} Based on the results in \cite{su2026}, we reformulate the expression of full and simplified PEB and VEB into optimization-oriented forms. {These metrics directly characterize position and velocity estimation accuracy.
    While the full version of PEB and VEB precisely capture the coupling between them and provide rigorous accuracy indicators, the simplified one benefits from better tractability and potential for low-complexity algorithms.}
    \item \textbf{Full CRLB-based beamforming design:} We consider a power-minimization BF design, 
    subject to communication SINR constraints, sensing PEB and VEB requirements, and power limitation per BS. 
    We then reformulate the optimization as an SDP problem by lifting the BF vectors into covariance matrices and applying SDR, where the constraints on full PEB and VEB are addressed by double Schur complements.
    We further prove the tightness of the SDR through the \ac{KKT} conditions. 
    The proposed SDP formulation benefits from full \ac{DoF} while producing high computation effort. Hence, it is expected to provide high-quality service in low-mobility or quasi-static scenarios. 
    \item \textbf{Simplified CRLB-based PA design:} Addressing the requirements for timeliness in practical wireless networks, we develop a low-complexity two-stage PA solution based on fixed BF vectors. 
    To reduce complexity, we employ simplified PEB and VEB for sensing constraints, which can be handled by \ac{SOCP}.
    Specifically, classical \ac{RZF} is adopted for communication BF, while \ac{NSP}-based sensing BF cancels the interference to communication \acp{UE}, enabling a decoupled two-stage design.
    In the first stage, the communication PA is handled by the \ac{SCA}-\ac{SOCP} procedure. 
    In the second stage, the sensing constraints are addressed by SOCP with the remaining power.
    The resulting optimization problem benefits from low computation complexity. 
    \item \textbf{Performance investigation:} The simulation results show that the proposed PA algorithm incurs only a negligible performance loss compared with the SDP-BF scheme, while the average computation time of 25\,ms reflects its high efficiency and applicability in dynamic environments.
    \item \textbf{Facilitating ISAC deployment:} The results show that stringent PEB and VEB requirements can be satisfied with only marginal sensing power. This implies that the sensing task of parameter estimation may be integrated into wireless networks with negligible impact on communication performance and total power consumption. This finding is expected to facilitate the deployment of ISAC in wireless networks.
\end{itemize}

The rest of this paper is organized as follows: Section~\ref{sec.model} introduces the system model. The formulas of CRLB are adapted for beamforming in Section~\ref{sec.crlb}. The beamforming design is presented in Section~\ref{sec.BF}, followed by the PA algorithm introduced in Section~\ref{sec.PA}. The simulation results are provided in Section~\ref{sec.simulation}. 
Section~\ref{sec.conclusion} summarizes this work.

\textit{Notations:}  $(\cdot)^{*}$, $(\cdot)^{T}$, $(\cdot)^{H}$, and $(\cdot)^{-1}$ denote the conjugate, transpose, conjugate transpose, and inverse, respectively. 
$a$ and $A$ represent scalars, $\mathbf{a}$ and $\mathbf{A}$ represent vectors and matrices, respectively. 
The diagonalization of vector $\mathbf{a}$ is written by $\mathrm{diag}$($\mathbf{a}$). 
$\text{tr}(\cdot)$ denotes the trace of a square matrix.
$\mathcal{R}\{\cdot\}$ implies the real-part operator.
$\mathcal{CN}(\mu,\sigma^2)$ denotes the complex normal distribution with mean and variance of $\mu$ and $\sigma^2$.
\vspace{-1em}

\section{System Model}\label{sec.model}
Fig.~\ref{fig:model} illustrates an example of the cooperative ISAC network, where the planar position and velocity are involved in the considered traffic scenario. 
The BSs are assumed to be centrally coordinated by a \ac{CU} and thus are fully synchronized and are aware of all the transmission signals. 
We assume that all the $N$ BSs are performing joint transmission to the $U$ single-antenna UEs and simultaneously transmitting dedicated sensing signals to the $Q$ targets. 
Each BS receives target-reflected sensing waveforms transmitted by itself for monostatic sensing and those transmitted by others for bistatic sensing.
The positions of the $n$-th BS and $q$-th target are represented in Cartesian coordinates $\mathbf{p}_n=[x_n,y_n]^T$ and $\mathbf{p}_q=[x_q,y_q]^T$, and the velocity of the target is represented by $\mathbf{v}_q=[v_{x,q},v_{y,q}]^T$. 
We assume all the \ac{Tx} and sensing \ac{Rx} antennas of BSs are \acp{ULA} with sizes of $N_\text{t}$ and $N_\text{r}$, respectively. The inter-element spacing is set to $d=\lambda/2$, where $\lambda$ denotes the wavelength.
The \ac{AoD} and \ac{AoA} denote the angle $\theta$ between the target \ac{LoS} and the Tx/Rx antenna normal, and $\varphi$ denotes the bearing angle, i.e., the angle between the target LoS and the $x$-axis. Further, we assume the perfect \ac{CSI} is known. For sensing, we consider the scenario of target estimation, 
where the BF and PA optimization adopt target parameters estimated from the previous frame or predicted by tracking, and the task is to optimize sensing accuracy by minimizing the CRLB.


\begin{figure}
    \centering
    \includegraphics[width=0.8\linewidth]{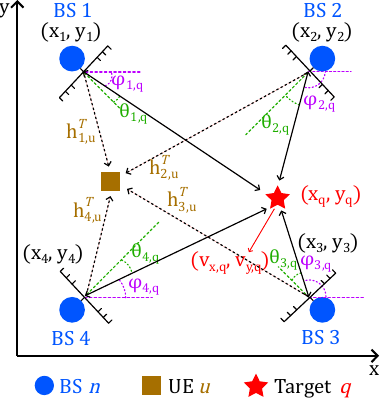}
    \caption{An example of the considered ISAC system model with JT for communication and MXS for sensing.}
    \label{fig:model}
\end{figure}

The transmitted signal from the $n$-th BS is given by
\begin{align}
    \mathbf{s}_n(t)=\sum_{u=1}^U \mathbf{w}_{n,u}{s}_u(t)+\sum_{q=1}^Q\mathbf{w}_{n,q}{s}_{n,q}(t),
\end{align}
where $\mathbf{w}_{n,u},\mathbf{w}_{n,q}\in\mathcal{C}^{N_\text{t}\times 1}$ are the communication and sensing BF vectors of BS $n$ for UE $u$ and target $q$, respectively, while the corresponding signals are denoted by ${s}_u(t)$ and ${s}_{n,q}(t)$, with the same power $P_s=0$\,dBW. The radar signals are orthogonal to each other and also to communication signals for ease of separation at the Rxs \cite{4408448}, i.e., 
\begin{align}\label{eq.orthogonal}
    &\int {s}_u(t)s_{n,q}^*(t-t_0)\,dt={0},\nonumber\ \int s_{n,q}(t)s_{n',q'}^*(t- t_0)\,dt=0,\\ 
    &\ \forall n,u,q,t_0,(n,q)\ne(n',q'),
\end{align}
which can be guaranteed by the zero-mean and independent characteristics of signals \cite{10380513}.

The received signal at the $u$-th UE is given by
\begin{align}\label{eq.resu}
    r_u(t)=&\sum_{n=1}^N\mathbf{h}_{n,u}^T\mathbf{w}_{n,u}{s}_u(t)+\sum_{n=1}^{N}\sum_{\substack{u'=1\\u'\ne u}}^{U}\mathbf{h}_{n,u}^T\mathbf{w}_{n,u'}{s}_{u'}(t)\nonumber\\
    &+\sum_{n=1}^{N}\sum_{q=1}^{Q}\mathbf{h}_{n,u}^T\mathbf{w}_{n,q}{s}_{n,q}(t)+n_u(t),
\end{align}
where $\mathbf{h}_{n,u}^T$ represents the channel parameter from BS $n$ to UE $u$. $n_u(t)\sim\mathcal{CN}(0,\sigma_u^2)$ denotes the \ac{AWGN} at UE $u$. The first term in (\ref{eq.resu}) is the desired signal, while the second and third terms represent the \ac{MUI} and \ac{SI}, respectively.
The SINR of the $u$-th UE is then calculated by
\begin{align}\label{eq.sinru}
    \gamma_u\negthinspace=\negthinspace\frac{\Big|\sum\limits_{n=1}^N\mathbf{h}_{n,u}^T\mathbf{w}_{n,u}\Big|^2}{\negthinspace\sum\limits_{\substack{u'=1\\u'\ne u}}^{U}\negthinspace\Big|\negthinspace\sum\limits_{n=1}^{N}\mathbf{h}_{n,u}^T\mathbf{w}_{n,u'}\Big|^2\negthinspace+\negthinspace\sum\limits_{n=1}^{N}\sum\limits_{q=1}^{Q}\negthinspace\Big|\mathbf{h}_{n,u}^T\mathbf{w}_{n,q}\Big|^2\negthinspace+\negthinspace\sigma_u^2}.\negthinspace
\end{align}

Assuming Gaussian signaling, the relationship between \ac{SE} and SINR can be represented by $\text{SE}_u=\log_2(1+\gamma_u)$. 
Since \ac{SE} is a monotonically increasing function of \ac{SINR}, maximizing or constraining \ac{SE} can be equivalently transformed into maximizing or constraining \ac{SINR}.
Therefore, we choose SINR as the communication metric in BF design.

The orthogonality defined in (\ref{eq.orthogonal}) eliminates the influence of the interference from communication streams and other sensing signals in the \ac{FIM} \cite{11570946}. Hence, by excluding communication interference and \ac{ITI}, 
the received sensing signal at the $m$-th BS is given by
\begin{align}
    &\mathbf{r}_{m}(t)
    =\nonumber\\
    &\negthinspace\sum_{n=1}^N\negthinspace\sum_{q=1}^Q\negthinspace{A_{n,m,q}\mathbf{a}_{m}^*\negthinspace(\negthinspace\theta_{m,q}\negthinspace)\mathbf{b}_{n}^{\negthinspace H}\negthinspace(\theta_{n,q})\mathbf{w}_{n,q}{s}_{n,q}(t\negthinspace-\negthinspace\tau_{n,m,q})e^{j2\pi f_{n,m,q}^\text{D}t}}\negthinspace\nonumber\\
    &+\mathbf{n}_{m}(t),\label{eq.radrec}
\end{align}
where the pair $(n,m,q)$ denotes the sensing link corresponding to Tx $n$, Rx $m$, and target $q$.
$A_{n,m,q}=a_{n,m,q}e^{j\phi_{n,m,q}}$ denotes the complex channel amplitude, $a_{n,m,q}=\sqrt{{\lambda^2\varsigma_{n,m,q} }/({(4\pi)^3(r_{n,q}r_{m,q})^2})}$, while $\varsigma_{n,m,q},\ r_{n,q},\text{ and }r_{m,q}$ represent the \ac{RCS} and ranges from Tx to target and from target to Rx in this link. $\mathbf{b}_{n}(\theta_{n,q})$ and $\mathbf{a}_{m}(\theta_{m,q})$ represent the Tx and Rx steering vectors, respectively. $\tau_{n,m,q}$ and $f_{n,m,q}^\text{D}$ are the delay and Doppler shift. $\mathbf{n}_m(t)\sim\mathcal{CN}(0,\sigma_m^2\mathbf{I}_{N_\text{r}})$ denotes the AWGN. In this work, we assume the noise power of each Rx is equal, i.e., $\sigma_m^2=\sigma_\text{s}^2,\forall m$. 

\section{CRLB of Networked ISAC}\label{sec.crlb}
We choose the CRLB as the sensing metric in this work, which is obtained by the inverse of the \ac{FIM}. 
The \ac{FIM} quantifies the amount of information carried by the received signals about the unknown target parameters, i.e., position $(x,y)$ and velocity $(v_x,v_y)$.
The discussion of FIM and CRLB in this work is based on two preconditions:
\begin{itemize}
    \item The multi-target coupling in the FIM is negligible, hence the FIMs can be constructed for each target independently. It is proved in \cite{11481148} that \textit{a sufficient condition for multiple-target CRLB approaching the corresponding approximate or single-target CRLB is that the delay or Doppler difference between the targets is larger than the corresponding resolution cell size}. This condition is generally satisfied. 
    For instance, for an FR3 waveform at $7.5$\,GHz with a bandwidth of 100\,MHz and frame length of 10\,ms, the range and velocity resolutions of monostatic sensing are 1.5\,m and 2\,m/s, being small enough.
    Consequently, the reflection of each sensing signal $s_{n,q}(t)$ only by the corresponding target $q$ is included in (\ref{eq.radrec}), and the FIM of target $q$ is constructed only based on reflected $s_{n,q}(t),\forall n$.
    \item For each target, the observation of each link $(n,m)$ can be treated as independent; hence, the global FIM can be calculated by the sum of the local FIMs. This condition holds with orthogonal Tx sensing waveforms and independent AWGN across Rxs, since on the one hand, it can be deduced from \cite{11570946} that \textit{for MIMO radar systems with orthogonal transmitted signals and AWGN, the observations of different transmitted signals at the same Rx are independent after matched filtering with the corresponding orthogonal waveforms}; on the other hand, the observations among different Rxs are independent due to the independent AWGN. These lead to independent observations for each link $(n,m)$.
\end{itemize}

In our recent work \cite{su2026}, we derived the full formula of the FIM. 
For target $q$, the FIM is given by
\begin{align}\label{eq.FIM}{
    \mathbf{F}_q=\sum_{n=1}^N\sum_{m=1}^M\mathbf{J}_{n,m,q}^T\mathbf{E}_{n,m,q}\mathbf{J}_{n,m,q}=\begin{bmatrix}
        \mathbf{F}_{\text{P},q} & \mathbf{F}_{\text{PV},q}\\
        \mathbf{F}_{\text{PV},q}^T & \mathbf{F}_{\text{V},q}
    \end{bmatrix},
}\end{align}
where $\mathbf{F}_{\text{P},q}$ and $\mathbf{F}_{\text{V},q}$ are the position and velocity blocks, $\mathbf{F}_{\text{PV},q}$ denotes the coupling between the position and velocity information, which is usually overlooked in related works \cite{10942860,10615966,10577579,8835674,10292936,9842350,10436719,11031778}. $\mathbf{F}_q$ is a real symmetric matrix. $\mathbf{E}_{n,m,q}$ expresses the \ac{EFIM} of the signal-level parameters (range $r$, radial velocity $v_\text{r}$, and AoA $\theta$). In our work \cite{11570946}, we have shown that for common ISAC waveforms like FMCW, sinc-PMCW, OFDM, and OTFS, we have
\begin{align}\label{eq.sige}
    \mathbf{E}_{n,m,q}\approx\text{diag}([E_{r,n,m,q},E_{v_\text{r},n,m,q}, E_{\theta,n,m,q}]),
\end{align}
where
\begin{align}\label{eq.siglvcrlb}
    &E_{r,n,m,q}=\frac{2\pi^2N_\text{r}\gamma_{n,m,q} B^2}{3c_0^2},\quad E_{v_\text{r},n,m,q}=\frac{2\pi^2N_\text{r}\gamma_{n,m,q} T_\text{F}^2}{3\lambda^2},\nonumber\\ &E_{\theta,n,m,q}=\frac{\pi^2\cos^2(\theta_{m,q})N_\text{r}(N_\text{r}^2-1)\gamma_{n,m,q}}{6},\nonumber\\
    &\gamma_{n,m,q}=\frac{a_{n,m,q}^2|\mathbf{b}_{n}^H(\theta_{n,q})\mathbf{w}_{n,q}|^2E_{s}}{\sigma^2_\text{s}},
\end{align}
where $c_0$ is the speed of light in vacuum.
$E_{s}=\sum_i|s_{n,q}[i]|^2=N_{s}P_s,\forall n,q$ denotes the energy of sampled signal, $N_s$ is the number of samples.
In this work, we assume all the BSs have the same frame length $T_\text{F}$, bandwidth $B$, and antenna size $N_\text{t},N_\text{r}$. 
$\mathbf{J}_{n,m,q}$ denotes the Jacobian transformation matrix:
\begin{align}
    \mathbf{J}_{n,m,q}&=\begin{bmatrix}
        \frac{\partial r_{n,m,q}}{\partial x} & \frac{\partial r_{n,m,q}}{\partial y} & \frac{\partial r_{n,m,q}}{\partial v_x} & \frac{\partial r_{n,m,q}}{\partial v_y}\\
        \frac{\partial v_{\text{r},n,m,q}}{\partial x} & \frac{\partial v_{\text{r},n,m,q}}{\partial y} & \frac{\partial v_{\text{r},n,m,q}}{\partial v_x} & \frac{\partial v_{\text{r},n,m,q}}{\partial v_y}\\
        \frac{\partial \theta_{m,q}}{\partial x} & \frac{\partial \theta_{m,q}}{\partial y} & \frac{\partial \theta_{m,q}}{\partial v_x} & \frac{\partial \theta_{m,q}}{\partial v_y}
    \end{bmatrix}.
\end{align}

The elements of $\mathbf{J}_{n,m,q}$ are calculated in the following. For the link $(n,m,q)$, the signal-level parameters are defined by
\begin{subequations}\begin{align}
    &{r_{n,m,q}}={r_{n,q}+r_{m,q}}={\Vert\mathbf{p}_q-\mathbf{p}_n\Vert+\Vert\mathbf{p}_q-\mathbf{p}_m\Vert}\nonumber\\
    &\negthinspace=\negthinspace{\sqrt{\Delta (x_{n,q})^{2}\negthinspace+\negthinspace(\Delta y_{n,q})^2}\negthinspace+\negthinspace\sqrt{\Delta (x_{m,q})^{2}\negthinspace+\negthinspace(\Delta y_{m,q})^2}},\negthinspace\\
    &{v_{\text{r},n,m,q}}=v_{\text{r},n,q}+v_{\text{r},m,q}=\mathbf{v}^T_q(\mathbf{u}_{n,q}+\mathbf{u}_{m,q})\nonumber\\
    &\negthinspace=\negthinspace v_{x,q}(\cos(\varphi_{n,q})\negthinspace+\negthinspace\cos(\varphi_{m,q}))\negthinspace+\negthinspace v_{y,q}(\sin(\varphi_{n,q})\negthinspace+\negthinspace\sin(\varphi_{m,q})),\\
    &\negthinspace\theta_{m,q}\negthinspace=\negthinspace\varphi_{m,q}\negthinspace-\negthinspace\varphi_{0,m}\negthinspace=\negthinspace\arctan({\Delta y_{m,q}}/{\Delta x_{m,q}})-\varphi_{0,m},\negthinspace
\end{align}\end{subequations}
where $\Delta x_{n,q}=x_q-x_n,\ \Delta y_{n,q}=y_q-y_n$. $\varphi_{0,m}$ represents the antenna normal direction of BS $m$. 
$\mathbf{u}_{n,q}=[\cos(\varphi_{n,q}),\sin(\varphi_{n,q})]^T=[c_{n,q},s_{n,q}]^T$ denotes the radial direction vector between BS $n$ and target $q$. The total range and radial velocity of the link $(n,m,q)$ are denoted by $r_{n,m,q}=r_{n,q}+r_{m,q}$ and $v_{\text{r},n,m,q}=v_{\text{r},n,q}+v_{\text{r},m,q}$.

The Jacobian matrix can be rewritten as
\begin{align}{
    \mathbf{J}_{n,m,q}=\begin{bmatrix}
        c_{n,m,q} & s_{n,m,q} & 0 & 0\\
        a_{n,m,q} & b_{n,m,q} & c_{n,m,q} & s_{n,m,q}\\
        -\frac{s_{m,q}}{r_{m,q}} & \frac{c_{m,q}}{r_{m,q}} & 0 & 0
    \end{bmatrix}\negthinspace,\negthinspace}
\end{align}
where 
\begin{align}\label{eq.mbscomp}
    &c_{n,m,q}=c_{n,q}+c_{m,q},\ s_{n,m,q}=s_{n,q}+s_{m,q},\nonumber\\
    &a_{n,m,q}=a_{n,q}+a_{m,q},\ b_{n,m,q}=b_{n,q}+b_{m,q},\nonumber\\
    &\mathbf{u}_{n,m,q}=\mathbf{u}_{n,q}+\mathbf{u}_{m,q},\ c_{n,q}=\frac{\Delta x_{n,q}}{r_{n,q}},\ s_{n,q}=\frac{\Delta y_{n,q}}{r_{n,q}},\nonumber\\
    &a_{n,q}\negthinspace=\negthinspace\frac{ v_{x,q}}{ r_{n,q}}-\frac{(\Delta x_{n,q}v_{x,q}+\Delta y_{n,q}v_{y,q})\Delta x_{n,q}}{(r_{n,q})^3}\negthinspace=\negthinspace-\frac{ s_{n,q}}{ r_{n,q}}v_{\perp n,q},\negthinspace\nonumber\\ 
    &b_{n,q}=\frac{ v_{y,q}}{ r_{n,q}}-\frac{(\Delta x_{n,q}v_{x,q}+\Delta y_{n,q}v_{y,q})\Delta y_{n,q}}{(r_{n,q})^3}=\frac{ c_{n,q}}{ r_{n,q}}v_{\perp n,q},\nonumber\\ 
    &v_{\perp n,q}=\mathbf{v}^T\mathbf{u}_{\perp n,q},\ \mathbf{u}_{\perp n,q}=[-s_{n,q},c_{n,q}]^T,\nonumber\\ &v_{\perp n,m,q}=v_{\perp n,q}+v_{\perp m,q},\ \mathbf{u}_{\perp n,m,q}=\mathbf{u}_{\perp n,q}+\mathbf{u}_{\perp m,q}.\negthinspace
\end{align}
In consequence, the elements of $\mathbf{F}_{\text{P},q}$, $\mathbf{F}_{\text{V},q}$, and $\mathbf{F}_{\text{PV},q}$ are given by
\begin{subequations}\label{eq.elements}\begin{align}
    &\mathbf{F}_{\text{P},q}=\sum_{n=1}^N|\mathbf{b}_n^H(\theta_{n,q})\mathbf{w}_{n,q}|^2\sum_{m=1}^N a_{n,m,q} ^2  \Big( E_{\text{r}}'\mathbf{U}_{n,m,q}\nonumber\\
    &+\frac{1}{(r_{m,q})^2}\cos^2(\theta_{m,q})E_{\theta}'\mathbf{U}_{\perp m,q}+E_{\text{v}_\text{r}}'\mathbf{G}_{n,m,q}\Big)\nonumber\\
    &=\sum_{n=1}^N|\mathbf{b}_n^H(\theta_{n,q})\mathbf{w}_{n,q}|^2\mathbf{R}_{\text{P},n,q},\\
    &\mathbf{F}_{\text{V},q}=\sum_{n=1}^N|\mathbf{b}_n^H(\theta_{n,q})\mathbf{w}_{n,q}|^2\sum_{m=1}^N a_{n,m,q} ^2  E_{\text{v}_\text{r}}'\mathbf{U}_{n,m,q}\nonumber\\
    &=\sum_{n=1}^N|\mathbf{b}_n^H(\theta_{n,q})\mathbf{w}_{n,q}|^2\mathbf{R}_{\text{V},n,q}\\
    &\mathbf{F}_{\text{PV},q}=\sum_{n=1}^N|\mathbf{b}_n^H(\theta_{n,q})\mathbf{w}_{n,q}|^2 \sum_{m=1}^N a_{n,m,q} ^2 E_{\text{v}_\text{r}}' \mathbf{Q}_{n,m,q}\nonumber\\
    &=\sum_{n=1}^N|\mathbf{b}_n^H(\theta_{n,q})\mathbf{w}_{n,q}|^2\mathbf{R}_{\text{PV},n,q},
\end{align}\end{subequations}
where 
\begin{align}
    &\negthinspace E_{{r}}'\negthinspace=\negthinspace\frac{2\pi^2\negthinspace N_\text{r}B^2E_{s}}{3c_0^2\sigma^2_\text{s}},\, E_{v_\text{r}}'\negthinspace=\negthinspace\frac{2\pi^2\negthinspace N_\text{r}T_\text{F}^2E_{s}}{3\lambda^2\sigma^2_\text{s}},\, E_{\theta}'\negthinspace=\negthinspace\frac{\pi^2\negthinspace N_\text{r}(N_\text{r}^2\negthinspace-\negthinspace 1)E_{s}}{6\sigma^2_\text{s}},\negthinspace\nonumber\\
    &\mathbf{U}_{n,m,q}=\mathbf{u}_{n,m,q}\mathbf{u}_{n,m,q}^T,\ 
    \mathbf{U}_{\perp m,q}=\mathbf{u}_{\perp m,q}\mathbf{u}_{\perp m,q}^T,\nonumber\\ 
    &\mathbf{G}_{n,m,q}=\mathbf{g}_{n,m,q}\mathbf{g}_{n,m,q}^T,\ 
    \mathbf{Q}_{n,m,q}=\mathbf{g}_{n,m,q}\mathbf{u}_{n,m,q}^T,\nonumber\\
    &\mathbf{g}_{n,m,q}=[a_{n,m,q},b_{n,m,q}]^T.
\end{align}
$E_{\text{r}}'$, $E_{v_\text{r}}'$, and $E_{\theta}'$ are constant and are equal for all the links, the matrices $\mathbf{U}_{n,m,q}$, $\mathbf{U}_{\perp m,q}$, $\mathbf{G}_{n,m,q}$, and $\mathbf{Q}_{n,m,q}$ carry the geometrical information. $a_{n,m,q}^2$ indicates the path loss. As mentioned in Section~\ref{sec.model}, these parameters are assumed to be known and are sorted in $\mathbf{R}_{\text{P},n,q}$, $\mathbf{R}_{\text{V},n,q}$, and $\mathbf{R}_{\text{PV},n,q}$. Therefore, the elements in (\ref{eq.elements}) show a linear relationship with the BF gain. However, the BF gains of different Txs have different information contributions in both amount and direction, which is not considered in conventional BF or PA optimizing the overall SNR.

The full \acp{EFIM} of position and velocity are obtained via the Schur complement by treating the other as the nuisance parameters:
\begin{subequations}\begin{align}\label{eq.fullefim}
    &\mathbf{P}_{q}=\mathbf{F}_{\text{P},q}-\mathbf{F}_{\text{PV},q}\mathbf{F}_{\text{V},q}^{-1}\mathbf{F}_{\text{PV},q}^T,\\ 
    &\mathbf{V}_{q}=\mathbf{F}_{\text{V},q}-\mathbf{F}_{\text{PV},q}^T\mathbf{F}_{\text{P},q}^{-1}\mathbf{F}_{\text{PV},q}.
\end{align}\end{subequations}

The full CRLBs of position and velocity are denoted by the squares of PEB and VEB:
\begin{align}\label{eq.crlb}
    &C_{\text{P},q}=\text{tr}(\mathbf{P}_{q}^{-1})=\text{PEB}^2_q,\ 
    C_{\text{V},q}=\text{tr}(\mathbf{V}_{q}^{-1})=\text{VEB}^2_q.
\end{align}

Addressing the high complexity of $\mathbf{P}_q$ and $\mathbf{V}_q$, in \cite{su2026}, we provided a simplified version. The theoretical analysis and simulation results demonstrate that the difference between the full and simplified CRLBs in value can be neglected in most cases. The simplified EFIMs are given by
\begin{subequations}\label{eq.simp}\begin{align}
    &\mathbf{P}_{q}'=\sum_{n=1}^N|\mathbf{b}_n^H(\theta_{n,q})\mathbf{w}_{n,q}|^2\sum_{m=1}^N a_{n,m,q} ^2 \Big( E_{\text{r}}'\mathbf{U}_{n,m,q}\nonumber\\
    &\qquad\ +\frac{1}{(r_{m,q})^2}\cos^2(\theta_{m,q})E_{\theta}'\mathbf{U}_{\perp m,q}\Big)\nonumber\\
    &\quad\ =\sum_{n=1}^N|\mathbf{b}_n^H(\theta_{n,q})\mathbf{w}_{n,q}|^2\mathbf{R}_{\text{P},n,q}',\\
    &\mathbf{V}_{q}'=\mathbf{F}_{\text{V},q}=\sum_{n=1}^N|\mathbf{b}_n^H(\theta_{n,q})\mathbf{w}_{n,q}|^2\mathbf{R}_{\text{V},n,q}
\end{align}\end{subequations}
where the simplified position and velocity information are decoupled by preserving only their respective information components and neglecting the coupling terms $\mathbf{F}_{\text{PV}}$. The corresponding simplified CRLBs are given by
\begin{align}
    \negthinspace C_{\text{P},q}'\negthinspace=\negthinspace\text{tr}((\mathbf{P}'_{q})^{-1})\negthinspace=\negthinspace(\text{PEB}_q')^2,\ 
    C_{\text{V},q}'\negthinspace=\negthinspace\text{tr}((\mathbf{V}_{q}')^{-1})\negthinspace=\negthinspace(\text{VEB}_q')^2.\negthinspace
\end{align}

In the next section, we develop an SDP-based BF algorithm based on the full CRLB, which serves as the benchmark due to its full \ac{DoF} and rigorous sensing performance constraint, and is expected to provide high-quality service in low-mobility or quasi-static scenarios. 
Addressing the complexity and practicality, in Section~\ref{sec.PA}, a low-complexity PA algorithm leveraging the simplified CRLB is proposed.

\section{SDP-Based Beamforming Design}\label{sec.BF}

We consider a power-minimization design, where the objective is to minimize the total transmission power (\ref{eq.oriopta}), subject to communication SINR constraints (\ref{eq.orioptb}), sensing PEB (\ref{eq.orioptc}) and VEB (\ref{eq.orioptd}) requirements, and power limitation per BS (\ref{eq.oriopte}). The optimization problem is formulated as
\begin{subequations}\label{eq.oriopt}\begin{align}
    \negthinspace(\text{P}1.1):\ &\min_{\mathbf{w}_{n,u},\mathbf{w}_{n,q}} \sum_{n=1}^N\negthinspace\Big(\negthinspace\sum_{u=1}^U ||\mathbf{w}_{n,u}||^2+\sum_{q=1}^Q||\mathbf{w}_{n,q}||^2\Big),\negthinspace\label{eq.oriopta}\\
    &\ \mathrm{s.t.}\ \gamma_u\ge \Gamma_u,\ \forall u,\label{eq.orioptb}\\
    &\ \quad \ \ C_{\text{P},q}\le \sigma_{\text{P},q}^2,\ \forall q,\label{eq.orioptc}\\
    &\ \quad \ \ C_{\text{V},q}\le \sigma_{\text{V},q}^2,\ \forall q,\label{eq.orioptd}\\
    &\ \quad \ \ \sum_{u=1}^U||\mathbf{w}_{n,u}||^2+\sum_{q=1}^Q||\mathbf{w}_{n,q}||^2\le P_\text{t},\ \forall n,\label{eq.oriopte}
\end{align}\end{subequations}
where $\Gamma_u$ denotes the required SINR of UE $u$, $\sigma_{\text{P},q}^2$ and $\sigma_{\text{V},q}^2$ represent the constraints on $\text{PEB}_q^2$ and $\text{VEB}_q^2$, $P_\text{t}$ expresses the identical maximum transmission power of each BS. 
Although the objective function and the per-BS power constraints are convex quadratic, this problem is non-convex due to the fractional SINR constraints and the CRLB constraints involving double matrix inversion operations. To obtain convexity, we adopt SDR and double Schur-complement transformations to recast (\ref{eq.oriopt}) into an SDP problem. 

Considering the coherent superposition at the UEs (\ref{eq.sinru}), we form the global BF vector $\mathbf{w}_u=[\mathbf{w}_{1,u}^T,...,\mathbf{w}_{N,u}^T]^T$ and channel $\mathbf{h}_u=[\mathbf{h}_{1,u}^T,...,\mathbf{h}_{N,u}^T]^T$ for each UE, and define their covariance matrices as $\mathbf{W}_u=\mathbf{w}_u\mathbf{w}_u^H,\ \mathbf{H}_u=\mathbf{h}_u^*\mathbf{h}_u^T,\ \mathbf{H}_{n,u}=\mathbf{h}_{n,u}^*\mathbf{h}_{n,u}^T$, and the BF matrices should satisfy $\mathbf{W}_u\succeq0$ and $\mathrm{rank}(\mathbf{W}_u)=1$. The covariance BF matrix of sensing $\mathbf{W}_{n,q}=\mathbf{w}_{n,q}\mathbf{w}_{n,q}^H$ should also satisfy $\mathbf{W}_{n,q}\succeq0$ and $\mathrm{rank}(\mathbf{W}_{n,q})\le1$. With lifted covariance matrices, the quadratic terms are transformed into linear trace expressions, while the CRLB constraints can be represented by \acp{LMI} through the Schur complement.

The objective function in (\ref{eq.oriopta}) is rewritten as
\begin{align}\label{eq.oriopta1}
    \sum_{u=1}^U \mathrm{tr}(\mathbf{W}_{u})+\sum_{n=1}^N\sum_{q=1}^Q \mathrm{tr}(\mathbf{W}_{n,q}).
\end{align}

The communication SINR constraint in (\ref{eq.orioptb}) is transformed into an affine inequality with respect to the lifted covariance matrices:
\begin{align}\label{eq.orioptb1}
    &\negthinspace\frac{1}{\Gamma\negthinspace_u}\negthinspace\mathrm{tr}(\mathbf{H}_u\negthinspace\mathbf{W}\negthinspace_{u})\negthinspace-\negthinspace\negthinspace\sum_{\substack{u'=1\\u'\ne u}}^U\negthinspace\mathrm{tr}(\mathbf{H}_u\negthinspace\mathbf{W}\negthinspace_{u'})\negthinspace-\negthinspace\negthinspace\sum_{n=1}^N\negthinspace\sum_{q=1}^Q\negthinspace \mathrm{tr}(\mathbf{H}_{n,u}\mathbf{W}\negthinspace_{n,q})\negthinspace\negthinspace\ge\negthinspace\negthinspace\sigma_{\negthinspace u}^2,\negthinspace
\end{align}

Similarly, the inequality in (\ref{eq.oriopte}) can be rewritten as
\begin{align}\label{eq.oriopte1}
    \sum_{u=1}^U\mathrm{tr}(\mathbf{D}_n\mathbf{W}_{u})+\sum_{q=1}^Q\mathrm{tr}(\mathbf{W}_{n,q})\le P_\text{t},
\end{align}
where $\mathbf{D}_n=\mathrm{diag}([\underbrace{0,...,0}_{(n-1)N_\text{t}},\underbrace{1,...,1}_{N_\text{t}},\underbrace{0,...,0}_{(N-n)N_\text{t}}])$ represents the block selector matrix.

The constraints of CRLB are handled in the following. The FIM elements in (\ref{eq.elements}) are rewritten by
\label{eq.elements1}\begin{align}
    &\negthinspace\mathbf{F}\negthinspace_{\text{P},q}\negthinspace=\negthinspace\sum_{n=1}^N\negthinspace\mathrm{tr}(\mathbf{B}_{n,q}\negthinspace\mathbf{W}\negthinspace_{n,q})\mathbf{R}\negthinspace_{\text{P},n,q},\mathbf{F}\negthinspace_{\text{V},q}\negthinspace=\negthinspace\sum_{n=1}^N\negthinspace\mathrm{tr}(\mathbf{B}_{n,q}\negthinspace\mathbf{W}\negthinspace_{n,q})\mathbf{R}_{\text{V},n,q},\negthinspace\nonumber\\
    &\mathbf{F}_{\text{PV},q}=\sum_{n=1}^N\mathrm{tr}(\mathbf{B}_{n,q}\mathbf{W}_{n,q})\mathbf{R}_{\text{PV},n,q},
\end{align}
where $\mathbf{B}_{n,q}=\mathbf{b}_{n}(\theta_{n,q})\mathbf{b}_{n}^H(\theta_{n,q})$ is the covariance matrix of the Tx steering vector. $\mathbf{F}_{\text{P},q}$, $\mathbf{F}_{\text{V},q}$, and $\mathbf{F}_{\text{PV},q}$ are the affine matrices of $\mathbf{W}_{n,q}$, but $\mathbf{P}_q$, $\mathbf{V}_q$ and $C_{\text{P},q}$, $C_{\text{V},q}$ are not. To obtain an SDP formulation, we introduce auxiliary matrices $\mathbf X_{\mathrm P,q}$, $\mathbf Z_{\mathrm P,q}$, $\mathbf X_{\mathrm V,q}$, and $\mathbf Z_{\mathrm V,q}$, which satisfy
\begin{subequations}\begin{align}
    &\negthinspace\mathbf{P}_{q}^{-1}\negthinspace\preceq\negthinspace \mathbf{X}_{\text{P},q}^{-1} \negthinspace\preceq\negthinspace\mathbf{Z}_{\text{P},q}\Rightarrow \text{tr}(\mathbf{P}_{q}^{-1})\negthinspace\le\negthinspace\text{tr}(\mathbf{X}_{\text{P},q}^{-1})\negthinspace\le\negthinspace \text{tr}(\mathbf{Z}_{\text{P},q})\negthinspace\le\negthinspace\sigma_{\text{P},q}^2,\negthinspace\\
    &\negthinspace\mathbf{V}_{q}^{-1}\negthinspace\preceq\negthinspace \mathbf{X}_{\text{V},q}^{-1} \negthinspace\preceq\negthinspace\mathbf{Z}_{\text{V},q}\Rightarrow \text{tr}(\mathbf{V}_{q}^{-1})\negthinspace\le\negthinspace\text{tr}(\mathbf{X}_{\text{V},q}^{-1})\negthinspace\le\negthinspace \text{tr}(\mathbf{Z}_{\text{V},q})\negthinspace\le\negthinspace\sigma_{\text{V},q}^2.\negthinspace
\end{align}\end{subequations}

The constraints $C_{\text{P},q}\le \sigma_{\text{P},q}^2$ and $C_{\text{V},q}\le \sigma_{\text{V},q}^2$ are equivalently represented by respective double Schur complements:
\begin{subequations}\label{eq.orioptcd1}\begin{align}
    & {\setlength{\arraycolsep}{2pt}\begin{bmatrix}
        \mathbf{F}_{\text{P},q}\negthinspace-\negthinspace\mathbf{X}_{\text{P},q} & \mathbf{F}_{\text{PV},q}\\ \mathbf{F}_{\text{PV},q}^T & \mathbf{F}_{\text{V},q}
    \end{bmatrix}\negthinspace\succeq\negthinspace 0,\ 
    \begin{bmatrix}
        \mathbf{Z}_{\text{P},q} & \mathbf{I} \\ \mathbf{I} & \mathbf{X}_{\text{P},q}
    \end{bmatrix}\negthinspace\succeq\negthinspace 0,\ \text{tr}(\mathbf{Z}_{\text{P},q})\negthinspace\le\negthinspace\sigma_{\text{P},q}^2},\negthinspace\\
    & {\setlength{\arraycolsep}{2pt}\begin{bmatrix}
        \mathbf{F}_{\text{V},q}\negthinspace-\negthinspace\mathbf{X}_{\text{V},q} & \mathbf{F}_{\text{PV},q}^T\\ \mathbf{F}_{\text{PV},q} & \mathbf{F}_{\text{P},q}
    \end{bmatrix}\negthinspace\succeq\negthinspace 0,\ 
    \begin{bmatrix}
        \mathbf{Z}_{\text{V},q} & \mathbf{I} \\ \mathbf{I} & \mathbf{X}_{\text{V},q}
    \end{bmatrix}\negthinspace\succeq\negthinspace 0,\ \text{tr}(\mathbf{Z}_{\text{V},q})\negthinspace\le\negthinspace\sigma_{\text{V},q}^2.\negthinspace}
\end{align}\end{subequations}

Hence, the CRLB constraints are transformed into \acp{LMI}.
By collecting the transformed formulas in (\ref{eq.oriopta1})-(\ref{eq.oriopte1}) and (\ref{eq.orioptcd1}), the resulting optimization problem is given by 
\begin{subequations}\begin{align}
    \negthinspace(\text{P}1.2):\ &\min_{\mathbf{W}_{u},\mathbf{W}_{n,q},\mathbf{X}_{\text{P},q},\mathbf{Z}_{\text{P},q},\mathbf{X}_{\text{V},q},\mathbf{Z}_{\text{V},q}} (\ref{eq.oriopta1}),\label{eq.oriopta2}\\
    &\ \mathrm{s.t.}\ (\ref{eq.orioptb1}),(\ref{eq.oriopte1}),(\ref{eq.orioptcd1}),\ \forall n,u,q,\label{eq.orioptb2}\\
    &\ \quad \ \ \mathbf{W}_{u},\mathbf{W}_{n,q},\mathbf{X}_{\text{P},q},\mathbf{Z}_{\text{P},q},\mathbf{X}_{\text{V},q},\mathbf{Z}_{\text{V},q}\negthinspace\succeq\negthinspace 0,\ \forall n,u,q,\negthinspace\label{eq.orioptc2}\\
    &\ \quad \ \ \mathrm{rank}(\mathbf{W}_u)=1,\ \forall u,\label{eq.orioptd2}\\
    &\ \quad \ \ \mathrm{rank}(\mathbf{W}_{n,q})\le1,\ \forall n,q.\label{eq.oriopte2}
\end{align}\end{subequations}
However, (P1.2) is still non-convex due to the rank constraints. To obtain convexity, we apply SDR by dropping the rank-1 constraints (\ref{eq.orioptd2})-(\ref{eq.oriopte2}), resulting in an SDP optimization problem, where the objective function is linear in the covariance BF matrices, and all the constraints are \acp{LMI} or linear affine inequalities. The SDP problem is reformulated as
\begin{align}\label{eq.P1.3}
    (\text{P}1.3):\ \text{The same as P1.2, without (\ref{eq.orioptd2})-(\ref{eq.oriopte2}).}
\end{align}
The tightness of SDR is characterized in the following theorem.
\begin{Theorem}\label{t.1}
    The SDR in (\ref{eq.P1.3}) is tight, and any optimal solution of $\mathbf{W}_u,\mathbf{W}_{n,q}$ satisfies $\mathrm{rank}(\mathbf{W}_u)=1,\mathrm{rank}(\mathbf{W}_{n,q})\le1,\forall n,u,q$.
\end{Theorem}
\begin{proof}
    See Appendix~\ref{sec.app1}.
\end{proof}

\textbf{Theorem~\ref{t.1}} implies that in cooperative ISAC networks, the constraints of coherent SINR and the proposed full CRLB cannot destroy the tightness of SDR. Since (P1.3) is a convex SDP, its global optimum can be obtained using $\mathrm{CVX}$ in MATLAB with a standard \ac{IPM}. (P1.3) can result in a globally convex solution with rank-1 $\mathbf{W}_u$ and $W_{n,q}$, $\forall n,u,q$, and the corresponding BF vectors can be recovered by eigenvalue decomposition.

\textit{Complexity analysis:} In this paper, $\mathcal{O}(\cdot)$ is utilized to count the computational complexity, and the calculations with complexity related to $\mathcal{O}(1)$ are ignored. The number of variables $n_\text{var}$ in optimization is on the order of $\mathcal{O}(UN^2N_\text{t}^2+QNN_\text{t}^2)$, corresponding to $U$ communication BF matrices and $NQ$ sensing BF matrices, while the amount of variables in the auxiliary matrices is of order $\mathcal{O}(Q)$ and thus is ignored. The complexity of solving the SDP problem using the standard \ac{IPM} is $\mathcal{O}(\sqrt{\nu}\log(1/\epsilon)(n_\text{var}^3+n_\text{var}^2\sum_k\nu_k^2+n_\text{var}\sum_k\nu_k^3))$ \cite{doi}, where $\epsilon$ is the prescribed solution accuracy and is related to the iteration number. 
$\nu_k$ denotes the dimension of the $k$-th constraint, 
$\nu=\sum_{k}\nu_k\approx\mathcal{O}(UNN_\text{t}+QNN_\text{t})$, 
being dominated by the \ac{PSD} constraints (\ref{eq.orioptc2}) on $\mathbf{W}_u$ and $\mathbf{W}_{n,q}$, while the auxiliary matrices have much smaller sizes and are thus ignored. The SINR (\ref{eq.orioptb1}) and power (\ref{eq.oriopte1}) constraints are affine scalar inequalities, while the LMI constraints of CRLB have fixed small sizes, thus their terms in $\nu$ are ignored.
$\sum_k\nu_k^2\approx\mathcal{O}(UN^2N_\text{t}^2+QNN_\text{t}^2)$, $\sum_k\nu_k^3\approx\mathcal{O}(UN^3N_\text{t}^3+QNN_\text{t}^3)$. If $UN\gg Q$, the complexity of (P1.3) with standard IPM approximates $\mathcal{O}(U^{3.5}N^{6.5}N_\text{t}^{6.5}\log(1/\epsilon))$, which is the general complexity of SDP problems.

\section{Two-Stage Power Allocation Design}\label{sec.PA}
The high complexity of SDP-based BF limits its applicability to low-dynamicity or quasi-static scenarios like indoor hotspots. On the contrary, PA only optimizes a single parameter per beam, extremely reducing the computational effort and the amount of information exchanged on the fronthaul, with a loss in \ac{DoF} and performance \cite{10742291}. 
In this section, we adopt a classical communication-centric design \cite{10494224}, where \ac{RZF} and \ac{NSP} are used for communication and sensing BF, respectively. NSP guarantees that the projection of sensing beams lies in the null space of the communication channels, thereby eliminating the SI at the UEs. Thus, the problem is further simplified as a two-stage PA algorithm, in which the communication performance is prioritized, while sensing is performed opportunistically using the residual power resource.
Further, we leverage the simplified CRLB introduced in our previous work \cite{su2026} and solve its constraint by \ac{SOCP} \cite{soc} with much lower complexity than SDR-SDP with lifted covariance matrices.

The communication RZF BF is given by
\begin{align}
    \widehat{\mathbf{W}}=\mathbf{H}^H(\mathbf{H}\mathbf{H}^H+\xi\mathbf{I})^{-1},
\end{align}
where $\mathbf{H}=[\mathbf{h}_1,...,\mathbf{h}_U]^T$ denotes the communication channel matrix. The $u$-th column of $\widehat{\mathbf{W}}$ is the BF vector for UE $u$:
\begin{align}\label{eq.combf}
    &\widehat{\mathbf{w}}_u=[\widehat{\mathbf{W}}]_{:,u}=[\widehat{\mathbf{w}}_{1,u}^T,...,\widehat{\mathbf{w}}_{N,u}^T]^T,\nonumber\\
    &\Bar{\mathbf{w}}_{n,u}=\frac{\widehat{\mathbf{w}}_{n,u}}{||\widehat{\mathbf{w}}_{n,u}||},\ {\mathbf{w}}_{n,u}=\sqrt{p_{n,u}}\Bar{\mathbf{w}}_{n,u},
\end{align}
where $\Bar{\mathbf{w}}_{n,u}$ is the normalized BF vector, ${\mathbf{w}}_{n,u}$ denotes the BF vector with allocated power $p_{n,u}$. RZF introduces a trade-off between the MUI and AWGN. In this work, the regularization factor is set to $\xi=\frac{U\sigma_u^2}{NP_\text{t}}$. RZF does not take the impact of SI into account since the NSP beamformer projects the sensing beams onto the null space of the communication channels:
\begin{align}\label{eq.nspsensing}
    &\Bar{\mathbf{w}}_{n,q}=\frac{\mathbf{P}_{\perp n}\mathbf{b}_n(\theta_{n,q})}{||\mathbf{P}_{\perp n}\mathbf{b}_n(\theta_{n,q})||},\ {\mathbf{w}}_{n,q}=\sqrt{p_{n,q}}\Bar{\mathbf{w}}_{n,q},
\end{align}
where $\mathbf{P}_{\perp n}=\mathbf{I}-\mathbf{H}_n^H(\mathbf{H}_n\mathbf{H}_n^H)^{-1}\mathbf{H}_n$ denotes the NSP matrix, and $\mathbf{H}_n=[\mathbf{h}_{n,1},...,\mathbf{h}_{n,U}]^T$ represents the channel matrix from BS $n$ to all UEs. 

\textit{Remark:} In the case of $\mathbf{h}_{n,u}\approx k\mathbf{b}_{n}(\theta_{n,q}),\exists\, n,u,q$, communication beam of UE $u$ can also provide sufficient illumination to target $q$, eliminating the need for the corresponding sensing beam. This case is not further addressed in this work.

\subsection{Communication PA Design}

The communication SINR and sensing CRLB constraints are structurally decoupled, while the coupling in power constraint per BS still exists. To further reduce the complexity, we
form a two-stage PA algorithm, where priority is assigned to communication capability since it is the fundamental function of a mobile network, and the remaining power is allocated for sensing.
The communication PA problem is formulated as
\begin{subequations}\label{eq.pa0}\begin{align}
    \negthinspace(\text{P2.1-C}):\ &\min_{p_{n,u}} \sum_{n=1}^N\sum_{u=1}^U p_{n,u},\label{eq.pa01}\\
    &\ \mathrm{s.t.}\ \gamma_u\ge \Gamma_u,\ \forall u,\label{eq.pa02}\\
    &\ \quad \ \ \sum_{u=1}^Up_{n,u}\le P_\text{t},\ \forall n.\label{eq.pa03}\\
    &\ \quad \ \ p_{n,u}\ge0, \forall n,u.
\end{align}\end{subequations}

The SINR constraint for UE $u$ without SI is rewritten as
\begin{align}\label{eq.sinrex}
    &\frac{1}{\Gamma_u}{\Big|\negthinspace\sum\limits_{n=1}^N\negthinspace\sqrt{p_{n,u}}\mathbf{h}_{n,u}^T\bar{\mathbf{w}}_{n,u}\Big|^2}\negthinspace\ge\negthinspace{\sum\limits_{\substack{u'=1\\u'\ne u}}^{U}\negthinspace\Big|\negthinspace\sum\limits_{n=1}^{N}\negthinspace\sqrt{p_{n,u'}}\mathbf{h}_{n,u}^T\Bar{\mathbf{w}}_{n,u'}\negthinspace\Big|^2\negthinspace+\negthinspace\sigma_u^2}.\negthinspace
\end{align}
To obtain affine SINR constraints, we introduce auxiliary amplitude variables $a_{n,u}=\sqrt{p_{n,u}}$. 
Further, we define the equivalent channel parameters $g_{n,u}=\mathbf{h}_{n,u}^T\Bar{\mathbf{w}}_{n,u}$ and $g_{n,u,u'}=\mathbf{h}_{n,u}^T\Bar{\mathbf{w}}_{n,u'}$. 
Consequently, (\ref{eq.sinrex}) can be modified as 
\begin{align}\label{eq.sinrsoc}
    &\sum\limits_{\substack{u'=1\\u'\ne u}}^{U}\Big|\sum\limits_{n=1}^{N}{a_{n,u'}} g_{ n,u,u'}\Big|^{ 2}+\sigma_u^2\le\frac{1}{\Gamma_u}{\Big|\sum\limits_{n=1}^N{a_{n,u}}g_{ n,u}\Big|^{ 2}}\nonumber\\
    &\negthinspace\Rightarrow \negthinspace\begin{Vmatrix}\negthinspace\begin{bmatrix}
        \Big\{\sum\limits_{n=1}^{N}{a_{n,u'}}g_{n,u,u'}\Big\}_{\substack{u'=1\\u'\ne u}}^U\\\sigma_u
    \end{bmatrix}\negthinspace\end{Vmatrix}\negthinspace\negthinspace\le\negthinspace\frac{1}{\sqrt{\Gamma_u}}\Big|\negthinspace\sum_{n=1}^N\negthinspace{a_{n,u}}g_{n,u}\Big|,\negthinspace
\end{align}
where $ \Big\{\sum\limits_{n=1}^{N}{a_{n,u'}}g_{n,u,u'}\Big\}_{\substack{u'=1\\u'\ne u}}^U$ denotes a column vector collecting all the corresponding components with $u'\ne u$.
(\ref{eq.sinrsoc}) has a similar form as the constraint in \ac{SOC} $||\mathbf{A}_i\mathbf{x}+\mathbf{b}_i||\le \mathbf{c}_i^T\mathbf{x}+d_i, \forall i$; however, the right-hand term with absolute value is non-linear. To handle this nonlinearity,
we employ \ac{SCA} and introduce an iteration process.
At the beginning of each iteration $i$, we define a phase compensator based on the solution of the last iteration $\phi^{(i)}_u=\angle(\sum_{n=1}^N{a_{n,u}^{(i-1)}}g_{n,u})$, so that $e^{-j\phi_u^{(i)}}\sum_{n=1}^N{a_{n,u}}g_{n,u}$ is initially real in this iteration. 
Then during the optimization solving (P2.1-C) in this iteration, we replace $|\sum_{n=1}^N{a_{n,u}}g_{n,u}|$ by its linear lower bound $\mathcal{R}\{e^{-j\phi_u^{(i)}}\sum_{n=1}^N{a_{n,u}}g_{n,u}\}$, resulting in an SOC constraint with more conservative solution.
{The solution of one iteration is also feasible for the next, and the overall power is monotonically non-increasing but lower-bounded by 0. This guarantees the convergence.}
The phase compensation is employed to accelerate the convergence. It ensures the affine lower bound is tight at the previous solution, preventing unnecessary conservatism around the current operating point and improving the stability of the SCA iteration. 
The phase compensator is updated using the newly obtained solution, and the process is repeated until convergence.

Consequently, (P2.1-C) can be modified as an SOCP problem:
\begin{subequations}\label{eq.pac}\begin{align}
    (\text{P2.2-C}):\ &\min_{p_{n},a_{n,u}} \sum_{n=1}^N p_{n},\label{eq.pa11}\\
    &\ \mathrm{s.t.}\ (\ref{eq.sinrsoc}),\ \text{with replaced right part},\ \forall u,\label{eq.pa12}\\
    &\ \quad \ \ p_{n}\le P_\text{t},\ \forall n,\label{eq.pa13}\\
    &\ \quad \ \ \sum_{u=1}^Ua_{n,u}^2\le p_{n},\ \forall n, \label{eq.pa14}\\
    &\ \quad \ \ p_{n}\ge0, a_{n,u}\ge0,\ \forall n,u, \label{eq.pa15}
\end{align}\end{subequations}
where $p_n=\sum_{u=1}^Up_{n,u}$ denotes the power of BS $n$ allocated to communication. The inequality $\sum_{u=1}^Ua_{n,u}^2\le p_{n}$ is tight since $p_n$ only appears in the target function and power constraints. Due to the SCA procedure, the processes from phase compensation to (P2.2-C) are implemented repeatedly. The iteration terminates until the objective value $\sum_{n=1}^Np_n$ converges.

\textit{Remark:} In the simulation, the SOCP problem is solved by the function $\mathrm{coneprog}$ in MATLAB, where the quadratic power constraints in (\ref{eq.pa14}) should be reformulated as
\begin{align}
    \begin{Vmatrix}
         \begin{bmatrix}
             \{2a_{n,u}\}_{u=1}^U\\p_{n}-1
         \end{bmatrix}
     \end{Vmatrix}\le p_{n}+1,\ \forall n.
\end{align}
In addition, since $\mathrm{coneprog}$ requires real number input, in (\ref{eq.sinrsoc}), for each $u'$, $\sum_{n=1}^{N}{a_{n,u'}}g_{n,u,u'}$ 
should be replaced by $[\sum_{n=1}^{N}{a_{n,u'}}\mathcal{R}\{g_{n,u,u'}\},\sum_{n=1}^{N}{a_{n,u'}}\mathcal{I}\{g_{n,u,u'}\}]^T$ in the code.

\subsection{Sensing PA Design}

After communication PA, the remaining power of BS $n$ for sensing is $P'_{\text{t},n}=P_\text{t}-p_n$. The design goal of sensing PA is
\begin{subequations}\label{eq.s0}\begin{align}
    (\text{P2.1-S}):\ &\min_{p_{n,q}} \sum_{n=1}^N\sum_{q=1}^Qp_{n,q},\label{eq.s01}\\
    &\ \mathrm{s.t.}\ C_{\text{P},q}\le \sigma_{\text{P},q}^2,\ \forall q,\label{eq.s02}\\
    &\ \quad \ \ C_{\text{V},q}\le \sigma_{\text{V},q}^2,\ \forall q,\label{eq.s03}\\
    &\ \quad \ \ \sum_{q=1}^Qp_{n,q}\le P'_{\text{t},n},\ \forall n.\label{eq.s04}\\
    &\ \quad \ \ p_{n,q}\ge0, \forall n,q.\label{eq.s05}
\end{align}\end{subequations}

To handle the tractability of CRLB, we employ the simplified EFIMs $\mathbf{P}_q'$ and $\mathbf{V}_q'$ in (\ref{eq.simp}) since they are affine matrices of power. (\ref{eq.s02}) and (\ref{eq.s03}) for target $q$ can be equivalently expressed as 
\begin{subequations}\label{eq.a}\begin{align}
    &C_{\text{P},q}'=\frac{\mathbf{P}_q^{\prime(1,1)}+\mathbf{P}_q^{\prime(2,2)}}{\mathbf{P}_q^{\prime(1,1)}\mathbf{P}_q^{\prime(2,2)}-(\mathbf{P}_q^{\prime(1,2)})^2}\le \sigma_{\text{P},q}^2,\\
    &C_{\text{V},q}'=\frac{\mathbf{V}_q^{\prime(1,1)}+\mathbf{V}_q^{\prime(2,2)}}{\mathbf{V}_q^{\prime(1,1)}\mathbf{V}_q^{\prime(2,2)}-(\mathbf{V}_q^{\prime(1,2)})^2}\le \sigma_{\text{V},q}^2,
\end{align}\end{subequations}
where $\mathbf{P}_q^{\prime(i,j)}$ and $\mathbf{V}_q^{\prime(i,j)}$ represent the $(i,j)$-th element of $\mathbf{P}_q'$ and $\mathbf{V}_q'$:
\begin{align}
    \mathbf{P}_q^{\prime(i,j)}=\sum_{n=1}^N p_{n,q}g_{\text{P},n,q}^{ (i,j)},\ \mathbf{V}_q^{\prime(i,j)}=\sum_{n=1}^N p_{n,q}g_{\text{V},n,q}^{ (i,j)},
\end{align}
where $g_{\text{P},n,q}^{ (i,j)}=|\mathbf{b}_n^H(\theta_{n,q})\Bar{\mathbf{w}}_{n,q}|^2\mathbf{R}_{\text{P},n,q}^{\prime(i,j)}$ and $g_{\text{V},n,q}^{ (i,j)}=|\mathbf{b}_n^H(\theta_{n,q})\Bar{\mathbf{w}}_{n,q}|^2\mathbf{R}_{\text{V},n,q}^{(i,j)}$ collect the constant elements. 
(\ref{eq.a}) can be equivalently transformed into the following \ac{SOC} constraints:

\begin{subequations}\label{eq.soca}\begin{align}
    \negthinspace&\begin{Vmatrix}\negthinspace
        \begin{bmatrix}
            2\sum\limits_{n=1}^N p_{n,q}g_{\text{P},n,q}^{ (1,2)}\\\negthinspace\sum\limits_{n=1}^N\negthinspace p_{n,q}(g_{\text{P}\negthinspace,n,q}^{ (1,1)}\negthinspace-\negthinspace g_{\text{P}\negthinspace,n,q}^{ (2,2)})\negthinspace\\2/\sigma_{\text{P},q}^2
        \end{bmatrix}\negthinspace
    \end{Vmatrix}\negthinspace\negthinspace\negthinspace\le\negthinspace \sum_{n=1}^N\negthinspace p_{n,q}(g_{\text{P}\negthinspace,n,q}^{ (1,1)}\negthinspace+\negthinspace g_{\text{P}\negthinspace,n,q}^{ (2,2)})\negthinspace-\negthinspace\frac{2}{\sigma_{\text{P},q}^2},\negthinspace\\ 
    \negthinspace&\begin{Vmatrix}\negthinspace
        \begin{bmatrix}
            2\sum\limits_{n=1}^N p_{n,q}g_{\text{V},n,q}^{ (1,2)}\\\negthinspace\sum\limits_{n=1}^N\negthinspace p_{n,q}(g_{\text{V}\negthinspace,n,q}^{ (1,1)}\negthinspace-\negthinspace g_{\text{V}\negthinspace,n,q}^{ (2,2)})\negthinspace\\2/\sigma_{\text{V},q}^2
        \end{bmatrix}\negthinspace
    \end{Vmatrix}\negthinspace\negthinspace\negthinspace\le\negthinspace \sum_{n=1}^N\negthinspace p_{n,q}(g_{\text{V}\negthinspace,n,q}^{ (1,1)}\negthinspace+\negthinspace g_{\text{V}\negthinspace,n,q}^{ (2,2)})\negthinspace-\negthinspace\frac{2}{\sigma_{\text{V},q}^2}.\negthinspace\negthinspace
\end{align}\end{subequations}
Consequently, (\ref{eq.s02}) and (\ref{eq.s03}) are transformed into SOCP constraints, which can be solved by $\mathrm{coneprog}$ in MATLAB. Unlike communication PA, since the EFIMs are real matrices, no additional separation of real and imaginary parts is required. The sensing PA optimization problem is rewritten as
\begin{align}\label{eq.Aopt}
\text{(P2.2-S): } & \text{The same as (P2.1-S), with replacement}\nonumber \\ & \text{of
(\ref{eq.s02}) and (\ref{eq.s03}) by (\ref{eq.soca})}.
\end{align}

\begin{algorithm}[t]
\caption{Two-Stage Power Allocation.}
\label{alg:1}
\begin{algorithmic}[1]
\REQUIRE $ \mathbf{h}_{n,u},\mathbf{b}_n(\theta_{n,q}), \mathbf{R}_{\text{P},n,q}',\mathbf{R}_{\text{V},n,q},\forall n,u,q$.

\STATE \textbf{Initialize:} $P_{\rm c}^{(0)}=0,\ i=0$. 

\STATE Compute communication RZF and sensing NSP BF vectors $\Bar{\mathbf{w}}_{n,u}$ (\ref{eq.combf}) and $\Bar{\mathbf{w}}_{n,q}$ (\ref{eq.nspsensing}).

\STATE Compute the effective communication channels
$g_{n,u},\,g_{n,u,u'}$.

\WHILE{$i< I$}
    \STATE $i\leftarrow i+1$.
    
    \STATE Update the phase compensator $\phi_u^{(i)}$.

    \STATE Solve the communication SOCP (P2.2-C) (\ref{eq.pac}).


    \STATE Set $P_{\rm c}^{(i)}=\sum_{n=1}^N p_{n}$.

    \IF{$\frac{|P_{\rm c}^{(i)}-P_{\rm c}^{(i-1)}|}
    {P_{\rm c}^{(i-1)}}<\epsilon_{\rm c}$}
        \STATE \textbf{break}
    \ENDIF
\ENDWHILE

\STATE Set
$P'_{\text{t},n}=P_\text{t}-p_{n}$.

\STATE Solve the sensing SOCP (P2.2-S) (\ref{eq.Aopt}).

\IF{(P2.2-S) is infeasible}
    \STATE Declare a sensing outage. 
\ENDIF

\RETURN $\{p_{n,u}\}$ and $\{p_{n,q}\}$ or outage indicator.

\end{algorithmic}
\end{algorithm}

Algorithm~\ref{alg:1} describes the process of the proposed PA algorithm. The RZF- and NSP-based communication and sensing BF with normalized power is first implemented. The communication PA is implemented prior to the sensing PA and is solved by the SCA-SOCP procedure, where in the $i$-th iteration, the phase compensation with $\phi_u^{(i)}$ is implemented to guarantee that the affine lower bound is tight at the initial point. The iteration terminates when the iteration number reaches $I$ or the result converges. 
The sensing PA is performed opportunistically using the residual power resource at the end. {If (P2.2-S) is infeasible, a sensing outage is declared, and only communication beams are transmitted. However, the simulation results in Section~\ref{sec.simulation} show that the required sensing power is much lower than that of communication, hence (P2.2-S) is generally feasible.}

\textit{Complexity analysis:} The complexity of RZF is $\mathcal{O}(U^2NN_\text{t}+U^3)\approx\mathcal{O}(U^2NN_\text{t})$ in the case of $NN_\text{t}\gg U$. The complexity of NSP is $\mathcal{O}(N(QN_\text{t}^2+U^2N_\text{t}+UN_\text{t}^2+U^3))\approx\mathcal{O}(NN_\text{t}^2(U+Q))$ in the case of $N_\text{t}\gg U$. The number of variables in (P2.2-C) is $n_\text{var}=N+NU$, with $U$ $U$-dimensional \ac{SOC} constraints, $N$ power constraints, and $2N+NU$ linear constraints. Hence, the complexity of the communication PA SCA-SOCP problem with the standard IPM solver is approximately given by $\mathcal{O}(I(N^{3.5}U^{3.5}+N^{2.5}U^{5.5})\log(1/\epsilon))$. 
By the same token, (P2.2-S) has complexity of $\mathcal{O}(\sqrt{(N+Q)}N^3Q^3\log(1/\epsilon))$. Among these processes, the communication PA is expected to dominate the complexity. The PA algorithm does not involve per-antenna element optimization or covariance matrix lifting, hence avoiding the corresponding high complexity in SDP-BF.

\textit{Remark:} Although the proposed SDP-BF and two-stage PA schemes are based on the CRLB of the MXS mode, 
they can also cooperate with the MMS and MBS modes by setting the FIM contributions of the inactive sensing links to zero.

\section{Simulation Results}\label{sec.simulation}

\begin{figure}
    \centering
    \includegraphics[width=0.75\linewidth]{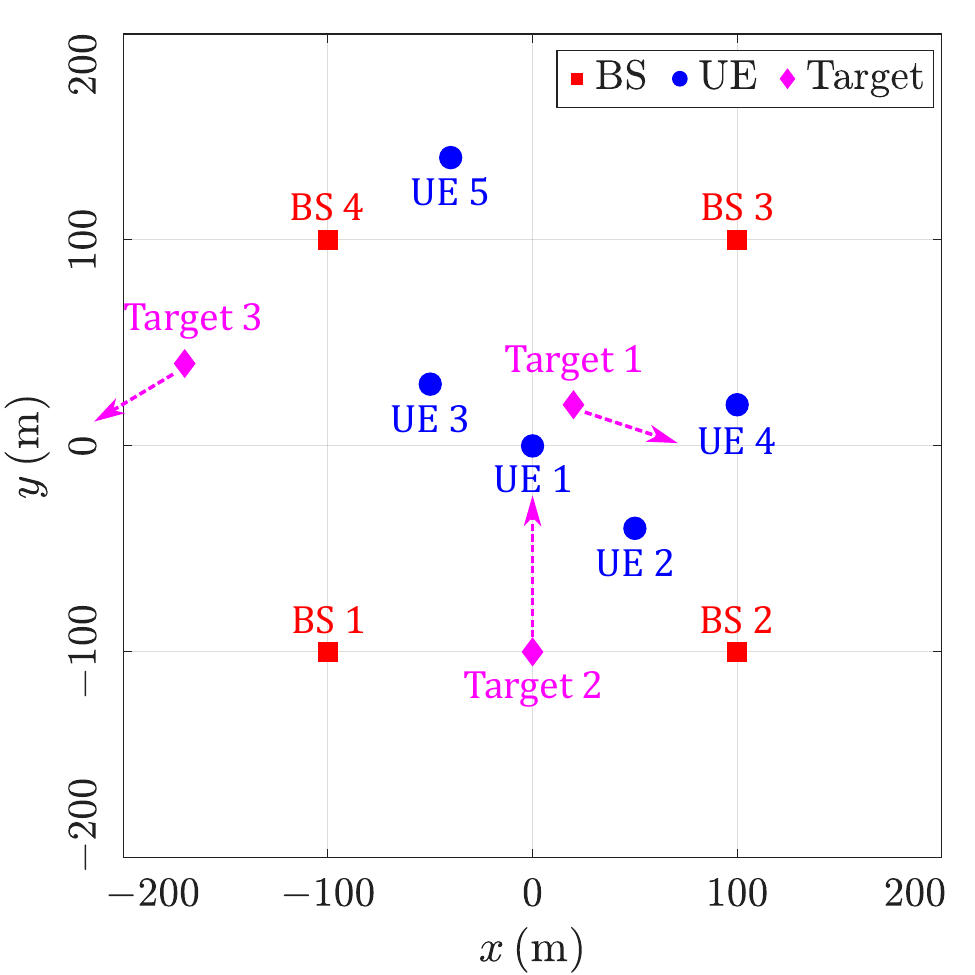}
    \caption{Distribution of BSs, UEs, and targets. The square formed by the four BSs is designated as the nominal service region. The target moving direction is represented by arrows.}
    \label{fig:map}
    \vspace{-1em}
\end{figure}

\begin{table}
    \centering
    \caption{System parameters.}
    \begin{tabular}{c|c}
        \hline
        Parameter & Value \\ \hline
        Number of BSs $N$ & 4 \\
        Number of UEs $U$ & 5 \\
        Number of targets $Q$ & 3 \\
        Carrier frequency $f_\text{c}$ & 7.5\,GHz \\
        Bandwidth $B$ & 100\,MHz \\
        Frame length $T_\text{F}$ & 10\,ms \\
        ULA size & $N_\text{t}=N_\text{r}=16$ \\
        Noise figure & UE: 9\,dB, Rx: 5\,dB \\
        Average RCS $\Bar{\varsigma}$ & 0\,dBsm \\
        SE constraint & 3\,bit/s/Hz \\
        PEB and VEB constraints $\sigma_\text{P}$ and $\sigma_\text{V}$ & 0.01\,m and 0.01\,m/s \\
        Power constraint per Tx $P_\text{t}$ & 35\,dBm \\
        \hline
    \end{tabular}
    \label{tab:para}
    \vspace{-1em}
\end{table}

In this section, we evaluate the behavior of the proposed SDP-BF and two-stage PA designs across different configurations and analyze their performance and time costs. We consider an \ac{UMi} scenario with 4 BSs, 5 UEs, and 3 targets. 
Without loss of generality, the sensing nodes are arranged at the vertices of a square (as shown in Fig.~\ref{fig:map}), their coordinates (in meters) are $\mathbf{p}_1=[-100,-100]^T,\mathbf{p}_2=[100,-100]^T,\mathbf{p}_3=[100,100]^T,$ and $\mathbf{p}_4=[-100,100]^T$, their ULAs are placed so that the normal direction points toward the origin $[0,0]^T$. The positions of UEs are set to $[0,0]$, $[50,-40]$, $[-50,30]$, $[100,20]$, and $[-40,140]$, and those of the targets are set to $[20,20]$, $[0,-100]$, and $[-170,40]$, with velocity (in km/h) of $[30,-10]$, $[0,50]$, and $[-25,-15]$. To illustrate the capability of the system more comprehensively, UE 4 and target 2 are located at the edge of the nominal service region, while UE 5 and target 3 are outside. {In \cite{su2026}, we demonstrated that the most significant numerical discrepancy between the full and simplified CRLBs occurs at the edges. Its impact can be observed by the tightness of the PEB and VEB constraints for target 2.}

Unless otherwise specified, the system configuration follows Table~\ref{tab:para}. The system operates in the FR3 band at 7.5\,GHz with a bandwidth of 100\,MHz and a frame length of 10\,ms, while the number of samples is calculated by $N_s=BT_\text{F}$ for multicarrier signals like OFDM and OTFS \cite{11570946}. 
The constraints of \acp{SE} are set to 3\,bit/s/Hz, corresponding to $\Gamma_u=\Gamma=7$\,dB, $\forall u$, while for sensing, we consider a stringent constraint with PEB and VEB below 0.01\,m and 0.01\,m/s to meet the requirement of high-accuracy sensing \cite{9815783,iturm2160}. Since the CRLB represents an optimistic lower bound that may not always be attained by practical estimators, such stringent thresholds also provide a conservative margin for practical sensing performance.
The \ac{NF}, power constraint per BS, and channel fading parameters are set according to 3GPP TR 38.901 \cite{3gpp38901}. 
We consider Rayleigh fading channels, where the channel vector between BS $n$ and UE $u$ follows $\mathbf{h}_{n,u}^T\sim\mathcal{CN}(0,{\beta_{n,u}}\mathbf{I}_{N_\text{t}})$, where $\beta_{n,u}=10^{-\text{PL}_{n,u}/10}$, and the path loss
\begin{align}
    \text{PL}_{n,u}=30.22+35.3\log_{10}(r_{n,u})+21.3\log_{10}(f\negthinspace_\text{c}),
\end{align}
in which $r_{n,u}$ is the distance between BS $n$ and UE $u$. For radar RCS, we consider the Swerling-I model \cite{Swerling}, where the RCS remains constant within one coherent processing interval (a frame) and varies independently between intervals. Its probability density function is given by $p(\varsigma_{n,m,q})=-\frac{1}{\Bar{\varsigma}}\exp(\frac{\varsigma_{n,m,q}}{\Bar{\varsigma}})$, where $\Bar{\varsigma}=\mathbb{E}\{\varsigma_{n,m,q}\}=1\,\text{m}^2\,\hat{=}\,0\,\text{dBsm}, \forall n,m,q$, is configured in the simulation. The \acp{NF} of Rxs and UEs are 5\,dB and 9\,dB, respectively. In the simulation, the BF (P1.3) and PA (P2.2-C \& P2.2-S) problems are solved by $\mathrm{CVX}$ and $\mathrm{coneprog}$ in MATLAB, respectively. 
The results in the simulation are averaged over 100 Monte Carlo realizations.

\subsection{Performance Analysis}

\begin{figure*}
    \centering
    \subfloat[Power vs. SINR constraint.]{
        \includegraphics[width=0.31\linewidth]{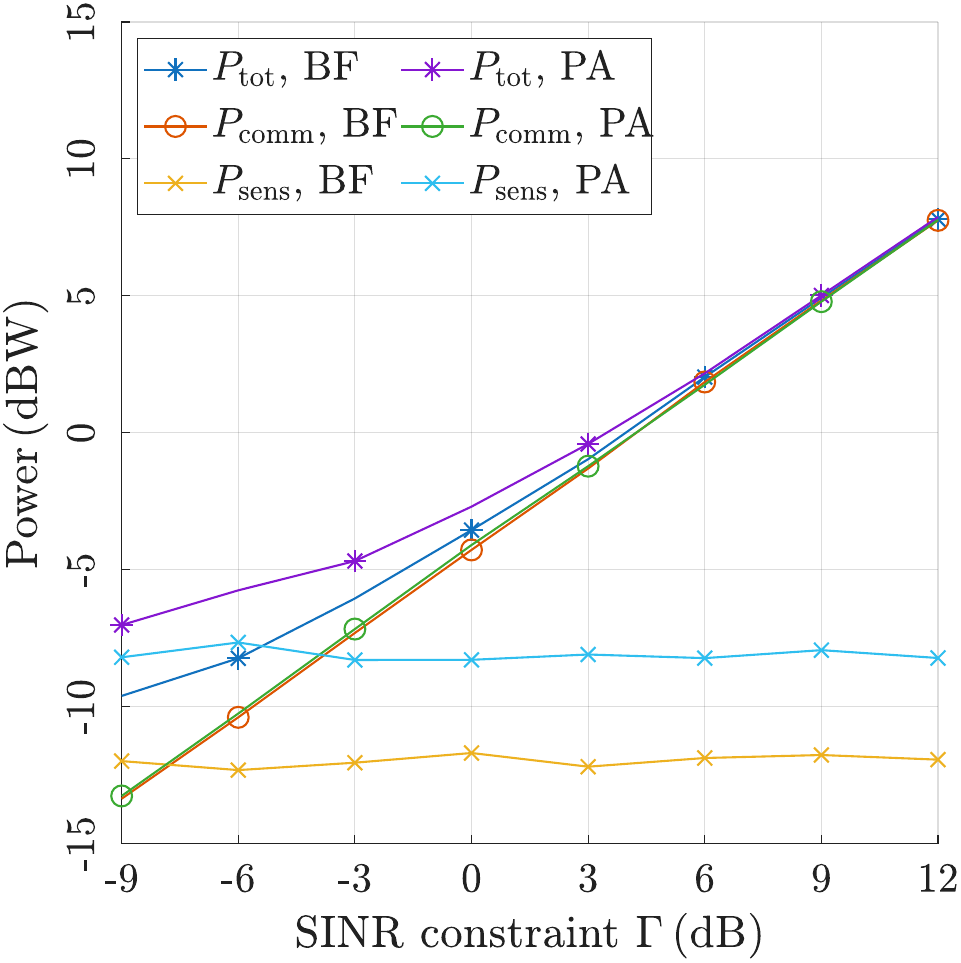}\label{fig.power1}
    }\hspace{-1mm}
    \subfloat[Power vs. PEB constraint.]{
        \includegraphics[width=0.31\linewidth]{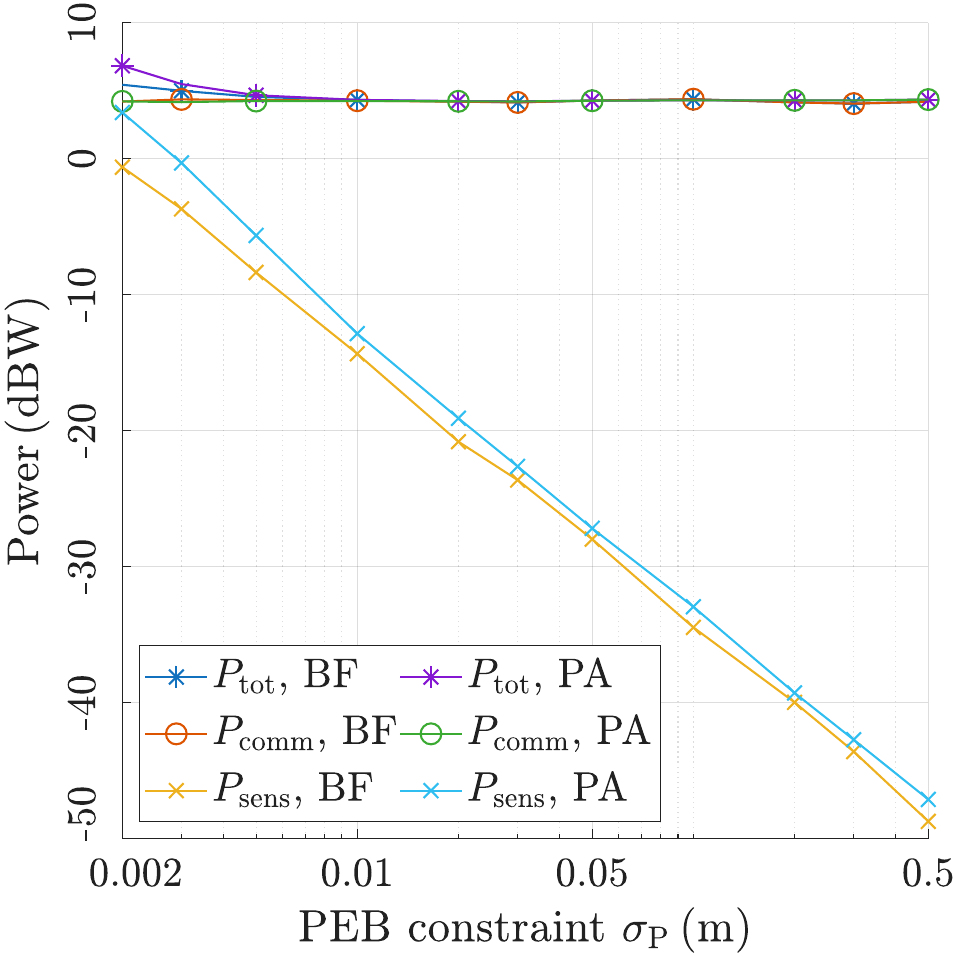}\label{fig.pow_peb1}
    }\hspace{-1mm}
    \subfloat[Power vs. VEB constraint.]{
        \includegraphics[width=0.31\linewidth]{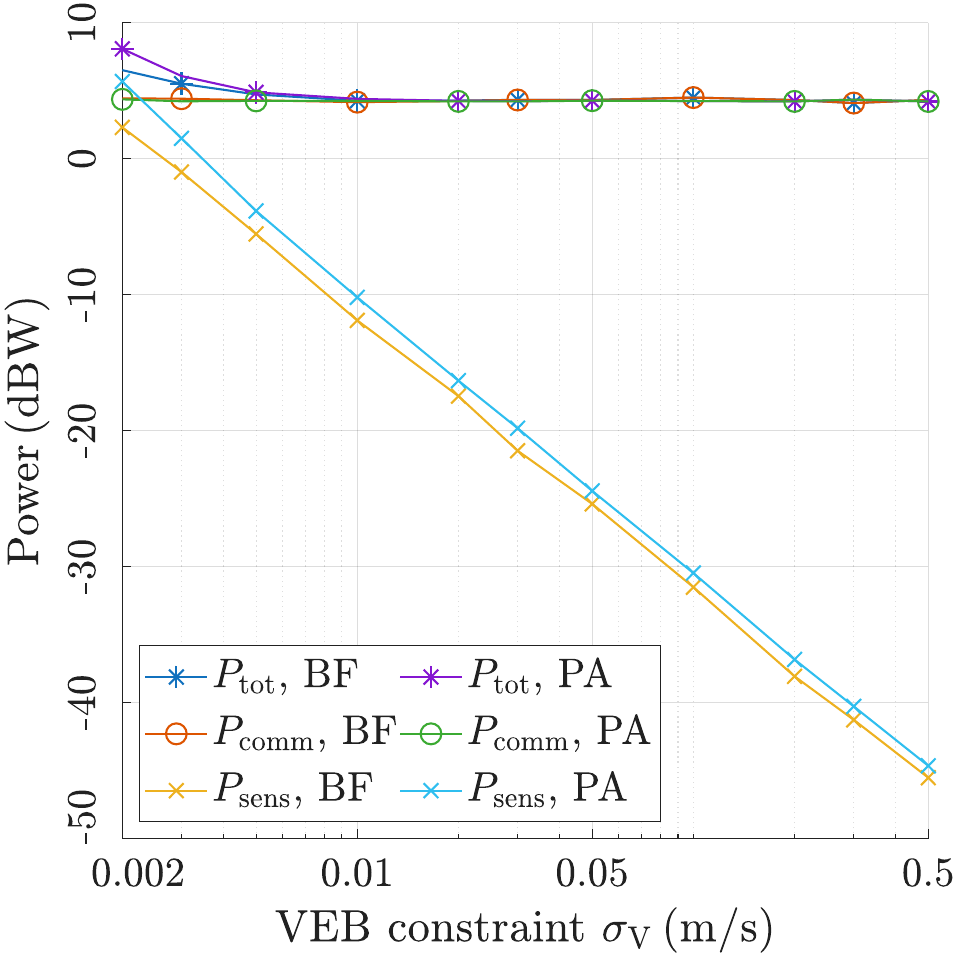}\label{fig.pow_veb1}
    }
    
    \subfloat[Measured average SINR vs. its constraint.]{
        \includegraphics[width=0.31\linewidth]{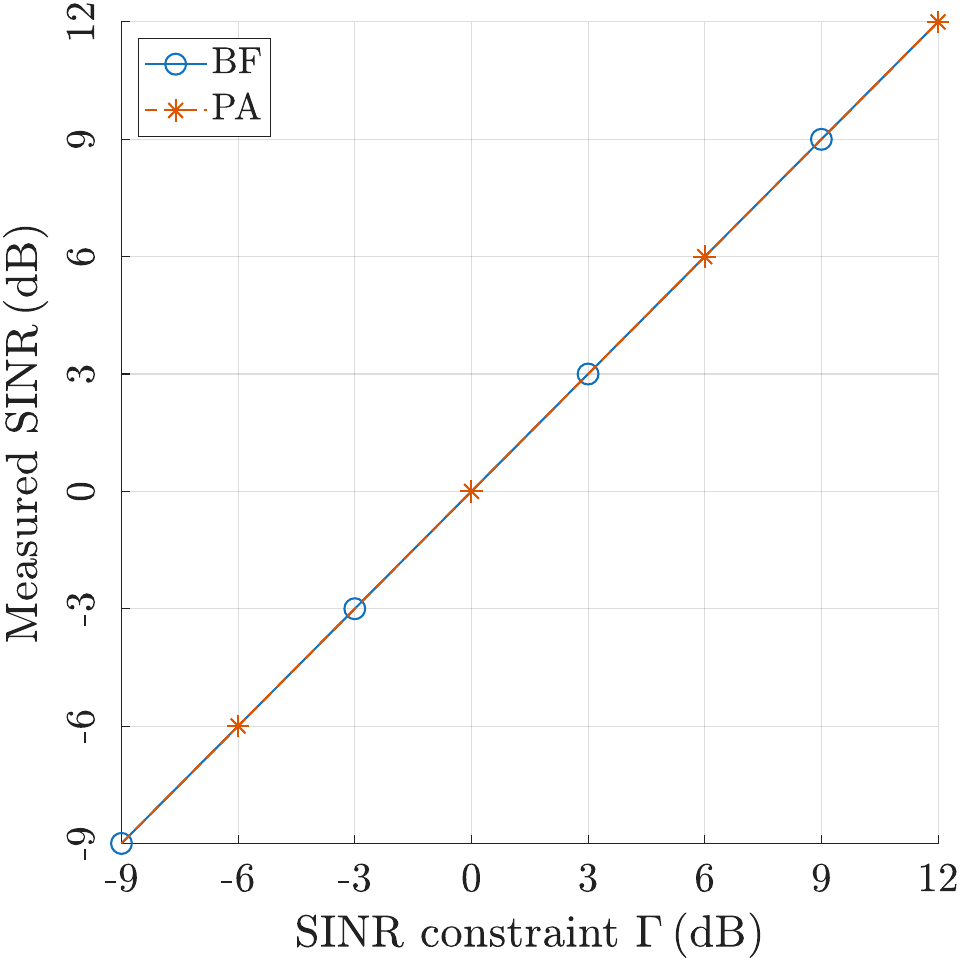}\label{fig.sinr}
    }\hspace{-1mm}
    \subfloat[Measured PEB vs. its constraint.]{
        \includegraphics[width=0.31\linewidth]{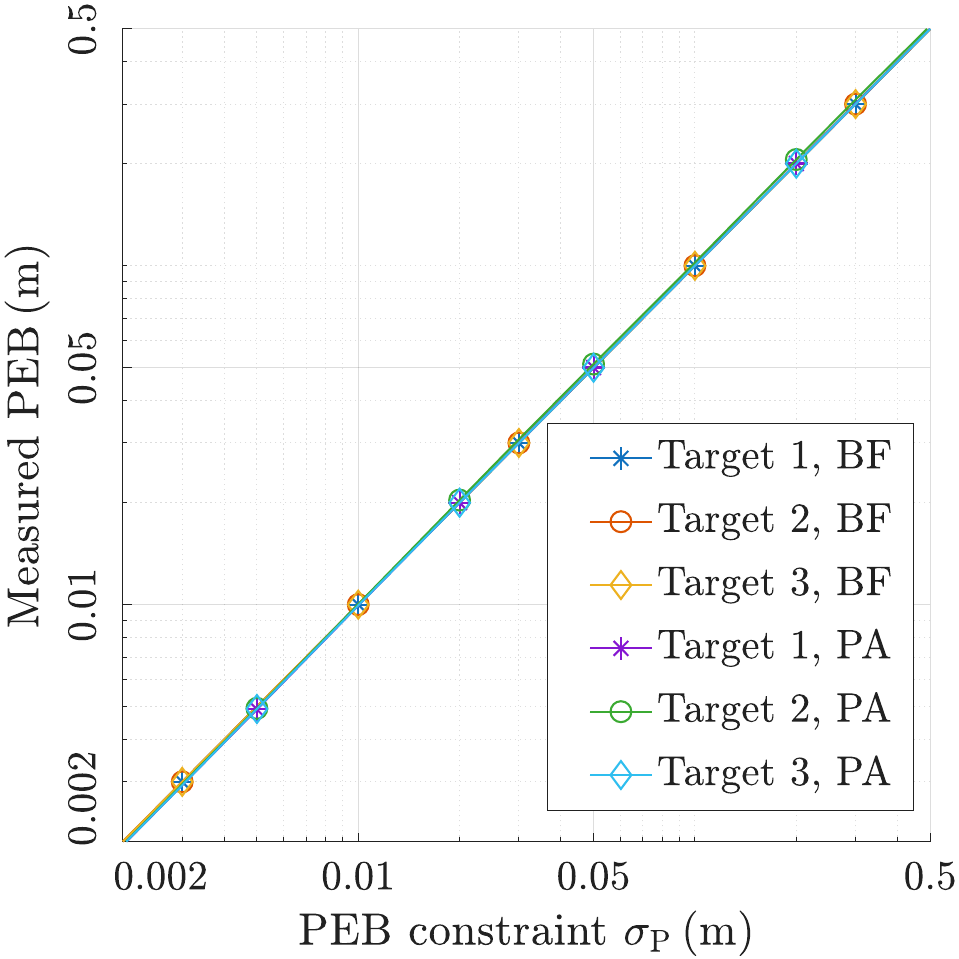}\label{fig.peb}
    }\hspace{-1mm}
    \subfloat[Measured VEB vs. its constraint.]{
        \includegraphics[width=0.31\linewidth]{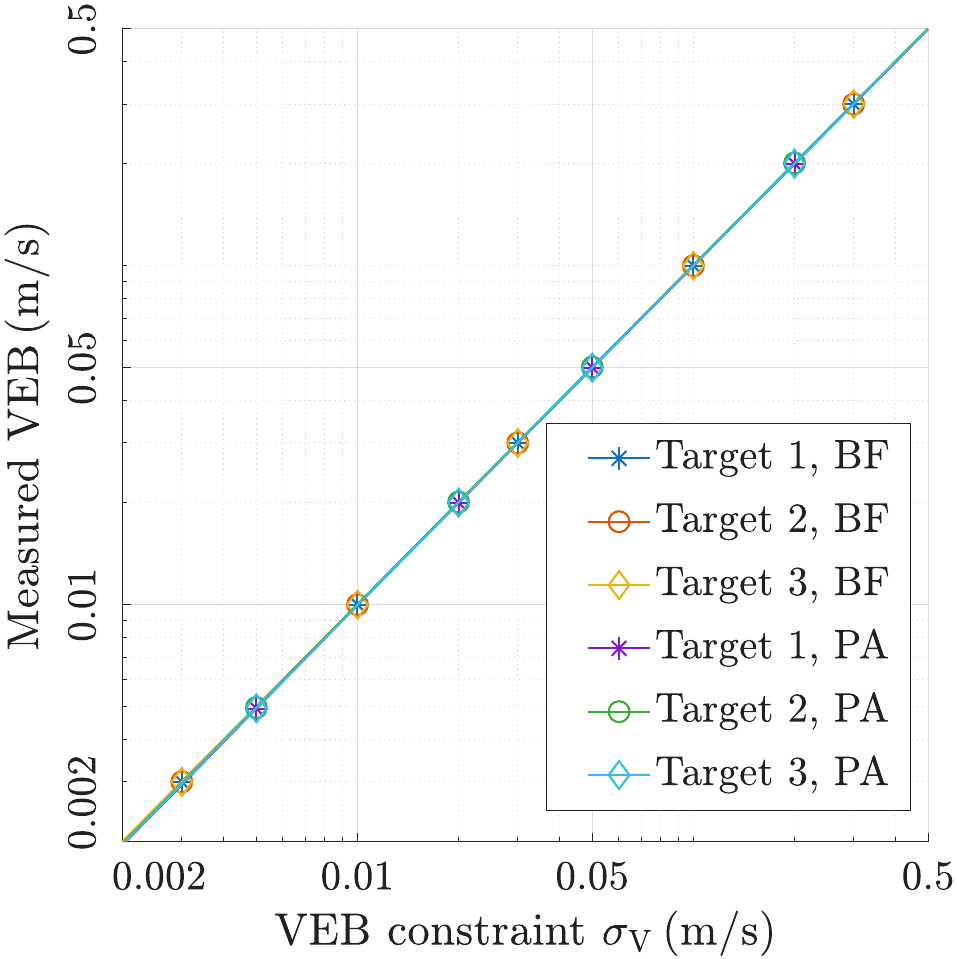}\label{fig.veb}
    }
    \caption{Impact of SINR and CRLB requirements on consumed power and performance of the proposed BF and PA schemes. $P_\text{tot}$, $P_\text{comm}$, and $P_\text{sens}$ denote the total power consumption and power for communication and sensing, respectively.}
    \label{fig.performance}
\end{figure*}

This section investigates the performance and behavior of the proposed SDP-BF and two-stage PA schemes with different settings. Fig.~\ref{fig.performance} illustrates the impact of SINR, PEB, and VEB constraints on power consumption and the satisfaction to the constraints. Since both $\mathbf{P}_q$ and $\mathbf{V}_q$ have a linear relationship with sensing beamforming gain, during the simulation of PEB and VEB constraints, the constraint on the other is relaxed. In addition, since the measured SINRs of UEs are identical and tightly close to the constraint in the simulation, only one representative SINR curve is shown for each of the two-stage PA and SDP-BF schemes in Fig.~\ref{fig.performance}(d) for visual clarity.

It is obvious that the communication streams are dominant in power consumption, while the allocated power for sensing is much smaller, being around $-8.3$\,dBW for two-stage PA and $-12.2$\,dBW for SDP-BF, although the requirement on sensing accuracy is extremely strict. 
Only when the required SINR is below $-3$\,dB (SE of 0.59\,bit/s/Hz) or the PEB and VEB constraints are stricter than 0.002\,m and 0.002\,m/s, the sensing power yields dominance. 

The easily satisfied CRLB constraint is caused by the large number of samples $N_s=BT_\text{F}=10^6$, providing up to 60 dB of radar processing gain.

Moreover, the constraints of SINR, PEB, and VEB are tightly satisfied by the SDP-BF and two-stage PA schemes, as illustrated in Fig.~\ref{fig.performance}(d)-(f). Although the communication PA procedure adopts an SCA process with a conservative treatment in each iteration, the SINR constraint is tightly satisfied with nearly identical communication power consumption to that of SDP-BF. Further, for targets 1, 2, and 3 located inside, at the edge of, and outside the nominal service area, the constraints on the PEB and VEB are tightly satisfied; this indicates that the numerical discrepancy between the simplified and full CRLBs has a negligible impact on sensing performance. 
However, the required sensing power of the PA scheme is around $1\sim 4$\,dB higher than SDP-BF, since the two-stage PA algorithm projects the sensing beam to the null space of the communication channels, resulting in sub-optimal sensing beamforming. Nevertheless, since the required sensing power is much less than communication, the difference between the proposed BF and PA schemes in total power consumption is tiny. 
A more detailed power distribution is provided in Fig.~\ref{fig.cc}. The SDP-BF and two-stage PA schemes yield nearly identical communication power distributions across the BSs, whereas the PA scheme allocates more sensing power at each BS. 

\begin{figure}
    \centering
    \subfloat[SDP-BF.]{
        \includegraphics[width=0.47\linewidth]{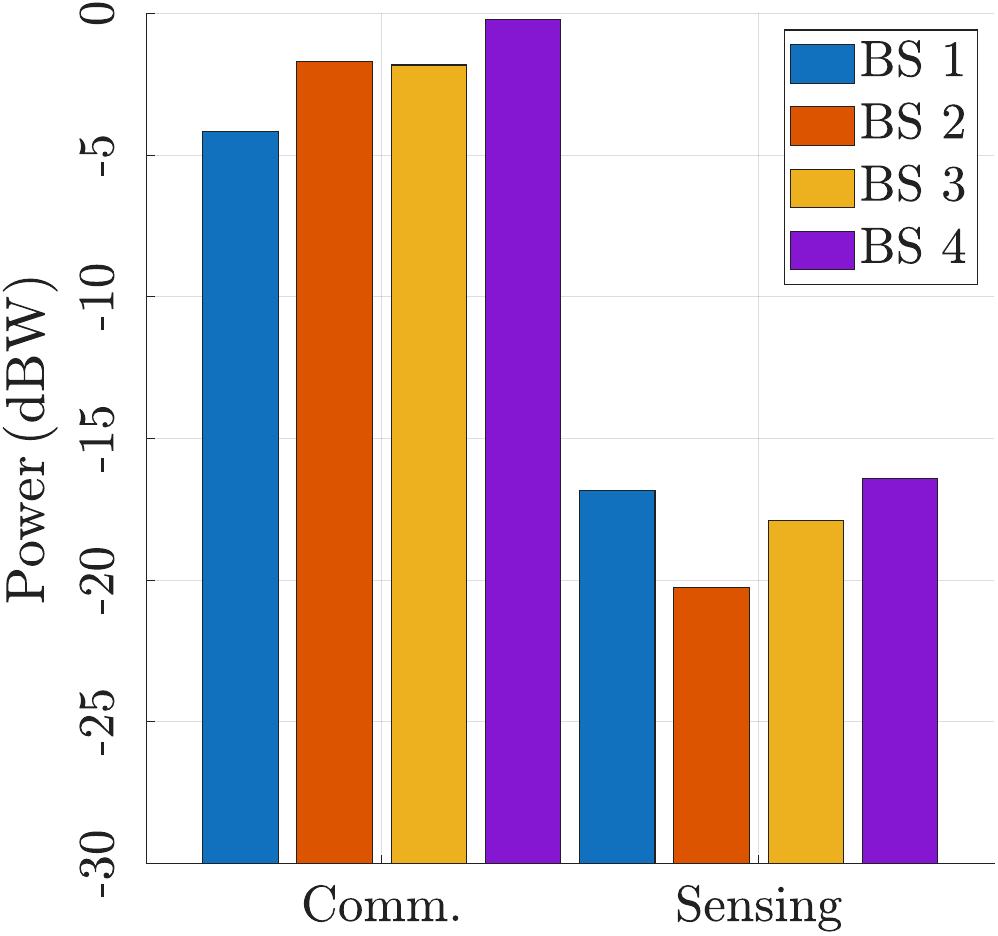}\label{fig.c1}
    }\hspace{-1mm}
    \subfloat[Two-stage PA.]{
        \includegraphics[width=0.47\linewidth]{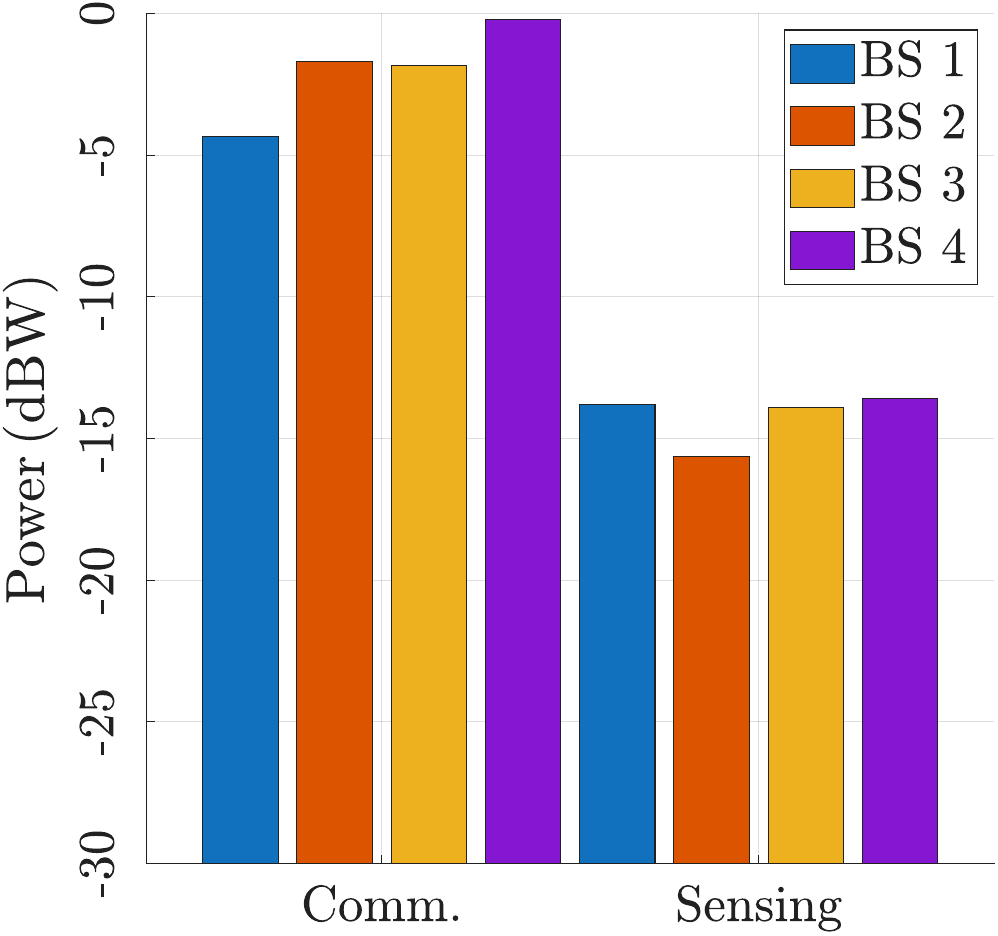}\label{fig.c2}
    }\hspace{-1mm}
    \caption{Power allocation among BSs and functionalities.}
    \label{fig.cc}
\end{figure}

The considered CRLB constraints characterize the scenario of post-detection parameter estimation accuracy, where the target is assumed to have been detected in the detection phase, and the BF or PA scheme is designed to address the position and velocity estimation accuracy. 
Figs.~\ref{fig.performance} and \ref{fig.cc} demonstrate that, 
\textit{under the considered scenario and from the perspective of parameter estimation and CRLB, the sensing performance is easy to satisfy, whereas guaranteeing the communication \ac{QoS} dominates the transmit power consumption. In this case, the integration of sensing may not lead to significant influence on communication performance or power consumption. This demonstrates the feasibility of deploying ISAC in 6G wireless networks.}
However, target detection \cite{10304081,10380513}, which is generally characterized by the detection probability and false alarm rate, may impose more stringent SNR requirements than target estimation and needs further investigation.

\subsection{Time Cost}

\begin{figure}
    \centering
    \subfloat[Time cost vs. $N_\text{t}$.]{
        \includegraphics[width=0.47\linewidth]{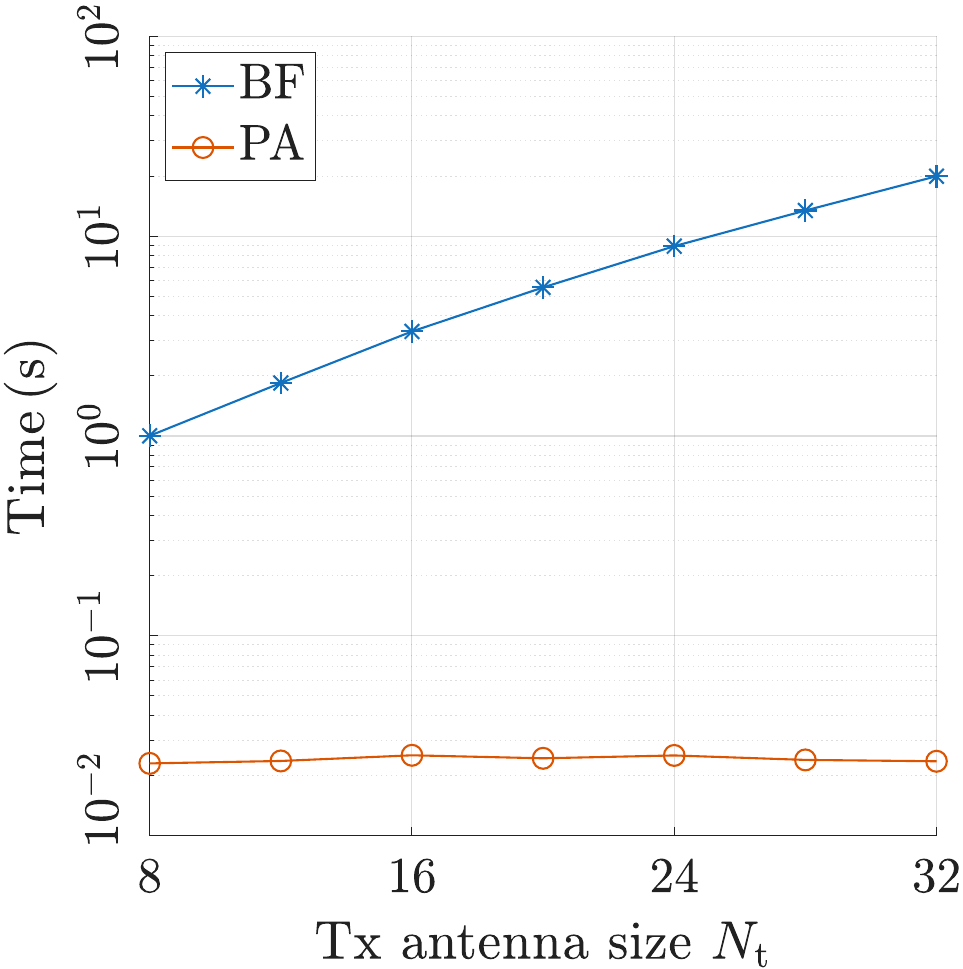}\label{fig.time_nt}
    }\hspace{-1mm}
    \subfloat[Time cost vs. $U,Q$.]{
        \includegraphics[width=0.47\linewidth]{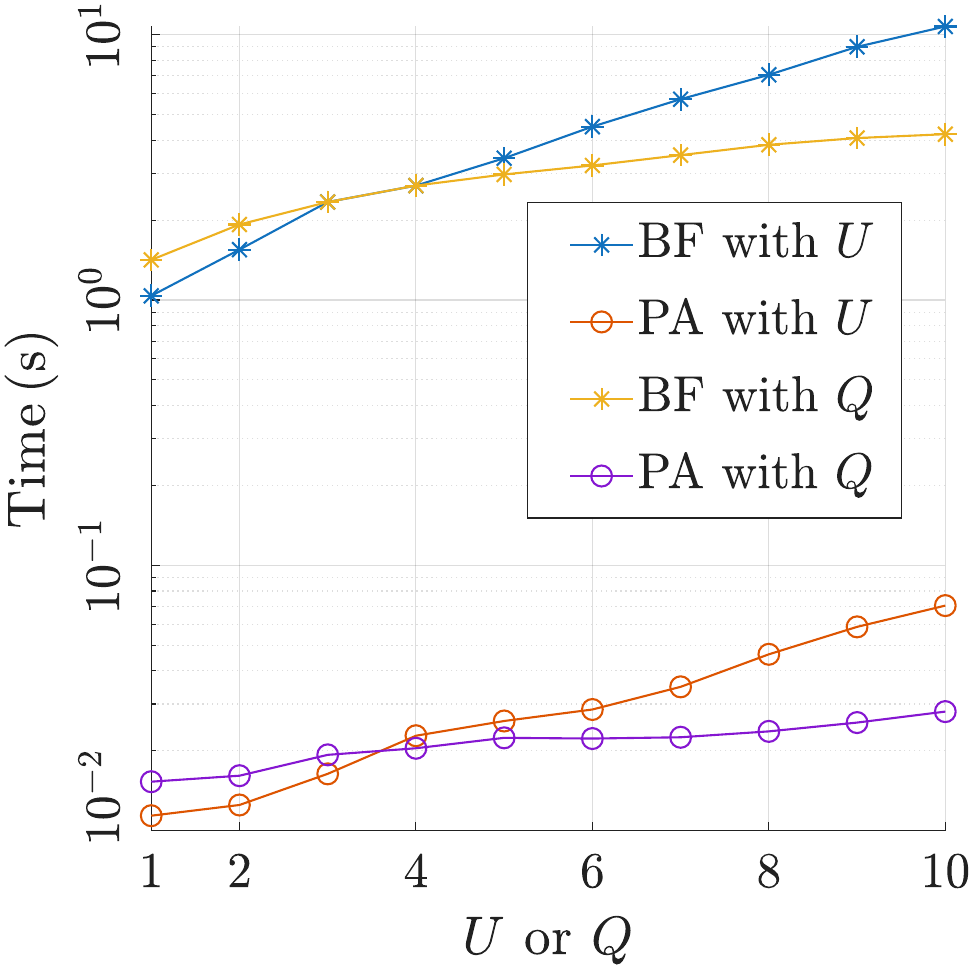}\label{fig.time_uq}
    }\hspace{-1mm}
    
    \subfloat[Time cost vs. $\Gamma$.]{
        \includegraphics[width=0.47\linewidth]{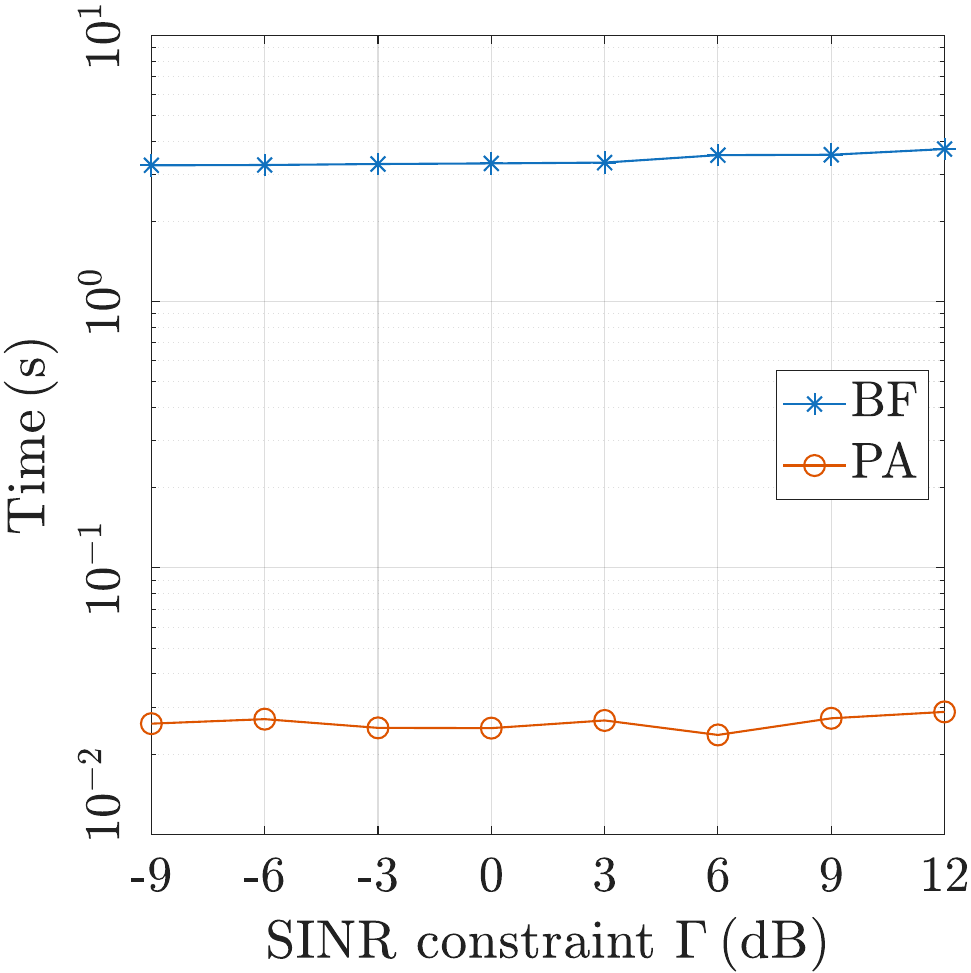}\label{fig.time_sinr}
    }\hspace{-1mm}
    \subfloat[Time cost vs. $\sigma_\text{P}\ \&\ \sigma_\text{V}$.]{
        \includegraphics[width=0.47\linewidth]{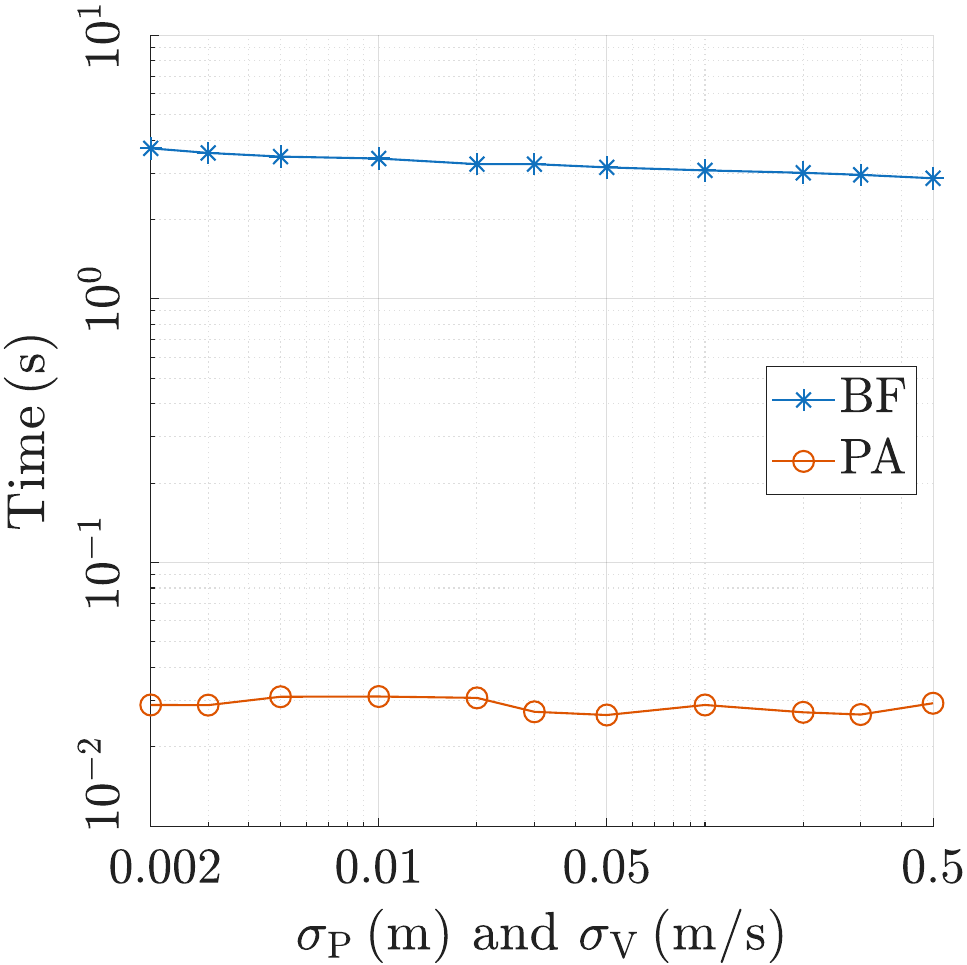}\label{fig.time_crlb}
    }
    \caption{Average computation time of the proposed SDP-BF and two-stage PA. During the tests of $U$ and $Q$, the other parameter is fixed to 3 for fairness, and the positions of the UEs and targets are generated randomly in the area depicted in Fig.~\ref{fig:map}.}
    \label{fig.time}
\end{figure}

Although the two-stage PA algorithm has a slight degradation in performance compared to the SDP-BF scheme, the complexity is significantly reduced. This section investigates their time efficiency. The time cost is recorded in MATLAB 2025b on an Intel Core i7-1165G7 processor.
The computational costs of the designed BF and PA schemes are measured by the time required to solve (P1.3) and to execute Algorithm~\ref{alg:1}, respectively. (P1.3) is solved using $\mathrm{CVX}$ with the $\mathrm{MOSEK}$ solver, while the PA problems (P2.2-C) and (P2.2-S) are solved by $\mathrm{coneprog}$. The recorded average time is illustrated in Fig.~\ref{fig.time}. The difference in time cost between the SDP-BF and two-stage PA schemes is significant. For the parameters given in Table~\ref{tab:para}, the computation time of SDP-BF is around 3.4\,s. In contrast, the time cost of the PA algorithm is around 25\,ms, which satisfies the requirements of high refreshing rate and low latency proposed in 3GPP TS 22.137 \cite{3gpp22137}, indicating its applicability in dynamic scenarios like traffic monitoring and tracking.

Further, since the PA process does not involve per-antenna element optimization and the classical RZF and NSP beamforming schemes are low-time-cost, the time cost of the PA algorithm does not change with antenna size. The number of UEs and targets can impact the computation time of both the SDP-BF and two-stage PA schemes, whereas the number of UEs has more significant influence than the number of targets, since $U$ is more dominant than $Q$ in complexity, as analyzed at the end of Sections~\ref{sec.BF} and \ref{sec.PA}. Additionally, the results show that the varied SINR and CRLB constraints cannot lead to obvious changes in the computation time.

\section{Conclusion}\label{sec.conclusion}
{To directly characterize and optimize sensing accuracy rather than surrogate SINR in optimization,
this paper provides a CRLB-based optimization framework for BF and PA in cooperative MIMO ISAC networks. 
In our design, we investigate how the proposed CRLB can be exploited and incorporated into system optimization.}
Specifically, we consider a power-minimization design, 
subject to the constraints of SINR, {PEB}, {VEB}, and per-{BS} power limitation. 
We first propose an {SDP}-based BF design, where the beamforming vectors are lifted as the covariance matrices, and the convexity is obtained by {SDR}. 
Particularly, the full PEB and VEB constraints are transformed into LMIs via double Schur complements. 
Addressing the complexity and timeliness, we further developed a two-stage PA algorithm, where the classical {RZF} and {NSP} methods are employed for communication and sensing beamforming, respectively, and simplified PEB and VEB are adopted.
The communication and sensing PA designs are then implemented sequentially and handled by SOCP.
The simulation result shows a negligible performance gap between the proposed BF and PA schemes, while the execution time of the PA algorithm is around 25\,ms, demonstrating its capability of real-time updating in dynamic scenarios. 
Furthermore, the results demonstrate that, for tasks of target estimation and from the aspect of CRLB, the integration of sensing in communication networks may not significantly influence the communication capability or power consumption. This conclusion is expected to accelerate the deployment of ISAC in wireless networks.


\begin{appendices}
\section{Proof of Theorem~\ref{t.1}}\label{sec.app1}
To demonstrate the tightness of the SDR, we analyze the \ac{KKT} conditions of the relaxed SDP (P1.3). 
Problem (P1.3) is convex and satisfies Slater’s condition \cite{Boyd}, following the standard treatment in related beamforming studies \cite{10742291,5755208}: since feasible solutions are generally with power margins due to the target function, the variable matrices can be perturbed toward the interior of the feasible set, yielding strict feasibility of other constraints.
Therefore, strong duality holds, and the KKT conditions \cite{Jarre2004} are necessary and sufficient for optimality.
In the following, we first show the tightness of the SDR for $\mathbf{W}_{u}$ with JT-type communication, and then address the sensing BF matrices $\mathbf{W}_{n,q}$.

\subsection{Communication Beamforming Matrix}
Let $\lambda_u\ge0$ denote the dual variable associated with the SINR constraint of UE $u$, $\rho_n\ge0$ denote that with the power constraint of BS $n$, and $\mathbf{Y}_{u}\succeq 0$ represent the dual matrix of the PSD constraint $-\mathbf{W}_{u}\preceq0$. The part of Lagrangian function of  (P1.3) related to $\mathbf{W}_{u}$ can be written as
\begin{align}
    \mathcal{L}(\mathbf{W}_{u})= &\mathrm{tr}(\mathbf{W}_{u})+\sum_{\substack{u'=1\\u'\ne u}}^U\lambda_{u'}\mathrm{tr}(\mathbf{H}_{u'}\mathbf{W}_{u})-\frac{\lambda_u}{\Gamma_u}\mathrm{tr}(\mathbf{H}_u\mathbf{W}_{u})\nonumber\\
    &+\sum_{n=1}^N\rho_n\text{tr}(\mathbf{D}_n\mathbf{W}_{u})-\text{tr}(\mathbf{Y}_{u}\mathbf{W}_{u}).
\end{align}

The stationary condition is obtained by taking $\frac{\partial\mathcal{L}}{\partial\mathbf{W}_{u}}=\mathbf{0}$, yielding
\begin{align}
    \negthinspace\mathbf{Y}_{u}^\text{opt}\negthinspace=\negthinspace\mathbf{I}\negthinspace+\negthinspace\sum_{\substack{u'=1\\u'\ne u}}^U\negthinspace\lambda_{u'}\mathbf{H}_{u'}\negthinspace-\negthinspace\frac{\lambda_u}{\Gamma_u}\mathbf{H}_u\negthinspace+\negthinspace\sum_{n=1}^N\rho_n\mathbf{D}_n\negthinspace=\negthinspace\mathbf{A}_u\negthinspace-\negthinspace\frac{\lambda_u}{\Gamma_u}\mathbf{H}_u,\negthinspace
\end{align}
where $\mathbf{Y}_{u}^\text{opt}$ denotes the optimal dual matrix corresponding to the PSD constraint of the optimal $\mathbf{W}_u^\text{opt}$.
$\mathbf{A}_u\succ 0$, $\mathrm{rank}(\mathbf{A}_u)=NN_\text{t}$, while $\mathrm{rank}(\mathbf{H}_u)=1$ since $\mathbf{H}_u=\mathbf{h}_u^*\mathbf{h}_u^T$. Thus, we have $\mathrm{rank}(\mathbf{Y}_u^\text{opt})\ge NN_\text{t}-1$ and $\dim(\mathcal{N}\{\mathbf{Y}_{u}^\text{opt}\})\le1$.

From the complementary slackness condition, we have $\text{tr}(\mathbf{Y}_{u}^\text{opt}\mathbf{W}_{u}^\text{opt})=0$. Considering the \ac{PSD} property, we have $\mathbf{Y}_{u}^\text{opt}\mathbf{W}_{u}^\text{opt}=\mathbf{0}$, implying that the columns of $\mathbf{W}_{u}^\text{opt}$ lie in the null space of $\mathbf{Y}_{u}^\text{opt}$, i.e., $\mathrm{Col}\{\mathbf{W}_{u}^\text{opt}\}\subseteq \mathcal{N}\{\mathbf{Y}_{u}^\text{opt}\}$, $\mathrm{rank}(\mathbf{W}_{u}^\text{opt})\le\dim(\mathcal{N}\{\mathbf{Y}_{u}^\text{opt}\})\le1$. Considering the required SINR $\Gamma_u>0$, $\mathbf{W}_{u}^\text{opt}\ne0$ should be guaranteed, hence $\mathrm{rank}(\mathbf{W}_{u}^\text{opt})=1$.

\subsection{Sensing Beamforming Matrix}

Let $\lambda_u\ge0$ and $\rho_{n}\ge0$ denote the dual variables associated with the SINR constraint of UE $u$ and power constraint of BS $n$, and 
$\mathbf{Y}_{n,q}\succeq0$ represent the dual matrix of the PSD constraint $-\mathbf{W}_{n,q}\preceq0$. Further introduce the dual matrices $\mathbf{Y}_{\text{P},q}\succeq0,\mathbf{Y}_{\text{V},q}\succeq0$ for the matrices $\mathbf{S}_{\text{P},q},\mathbf{S}_{\text{V},q}$ in the Schur complement (\ref{eq.orioptcd1}):
\begin{align}
    \negthinspace&{\setlength{\arraycolsep}{2pt} \mathbf{S}_{\text{P},q}=\begin{bmatrix}
        \mathbf{F}_{\text{P},q}-\mathbf{X}_{\text{P},q} & \mathbf{F}_{\text{PV},q}\\ \mathbf{F}_{\text{PV},q}^T & \mathbf{F}_{\text{V},q}
    \end{bmatrix},\ \mathbf{Y}_{\text{P},q}=\begin{bmatrix}
        \mathbf{Y}_{\text{P},q}^{11} & \mathbf{Y}_{\text{P},q}^{12}\\ (\mathbf{Y}_{\text{P},q}^{12})^T & \mathbf{Y}_{\text{P},q}^{22}
    \end{bmatrix},}\nonumber\\
    &{\setlength{\arraycolsep}{2pt} 
    \negthinspace\mathbf{S}_{\text{V},q}\negthinspace=\negthinspace\begin{bmatrix}
        \mathbf{F}_{\text{V},q}-\mathbf{X}_{\text{V},q} & \mathbf{F}_{\text{PV},q}^T\\ \mathbf{F}_{\text{PV},q} & \mathbf{F}_{\text{P},q}
    \end{bmatrix}\negthinspace\negthinspace,\ \mathbf{Y}_{\text{V},q}\negthinspace=\negthinspace\begin{bmatrix}
        \mathbf{Y}_{\text{V},q}^{11} & \mathbf{Y}_{\text{V},q}^{12}\\ (\mathbf{Y}_{\text{V},q}^{12})^T & \mathbf{Y}_{\text{V},q}^{22}
    \end{bmatrix}\negthinspace\negthinspace,\negthinspace}
\end{align}
where the elements $\mathbf{Y}_{\text{P},q}^{ij}$ and $\mathbf{Y}_{\text{V},q}^{ij}$ are $2\times2$ blocks. 
The trace of the multiplication of the Schur complement matrices and their dual is given by
\begin{subequations}\begin{align}
    \negthinspace&\text{tr}(\mathbf{Y}_{\text{P},q}\mathbf{S}_{\text{P},q})\negthinspace=\text{tr}(\mathbf{Y}_{\text{P},q}^{11}\mathbf{F}_{\text{P},q})-\text{tr}(\mathbf{Y}_{\text{P},q}^{11}\mathbf{X}_{\text{P},q})+\text{tr}(\mathbf{Y}_{\text{P},q}^{12}\mathbf{F}_{\text{PV},q}^T)\nonumber\\
    &\qquad\qquad\qquad\ +\text{tr}((\mathbf{Y}_{\text{P},q}^{12})^T\mathbf{F}_{\text{PV},q})+\text{tr}(\mathbf{Y}_{\text{P},q}^{22}\mathbf{F}_{\text{V},q}),\nonumber\\
    &\text{tr}(\mathbf{Y}_{\text{P},q}\mathbf{S}_{\text{P},q})|_{\mathbf{W}_{n,q}}=\mathrm{tr}(\mathbf{B}_{n,q}\mathbf{W}_{n,q})\eta_{\text{P},n,q},
    \\
    \negthinspace&\text{tr}(\mathbf{Y}_{\text{V},q}\mathbf{S}_{\text{V},q})\negthinspace=\negthinspace\text{tr}(\mathbf{Y}_{\text{V},q}^{11}\mathbf{F}_{\text{V},q})\negthinspace-\negthinspace\text{tr}(\mathbf{Y}_{\text{V},q}^{11}\mathbf{X}_{\text{V},q})\negthinspace+\negthinspace\text{tr}(\mathbf{Y}_{\text{V},q}^{12}\mathbf{F}_{\text{PV},q})\negthinspace\nonumber\\
    &\qquad\qquad\qquad\ +\text{tr}((\mathbf{Y}_{\text{V},q}^{12})^T\mathbf{F}_{\text{PV},q}^T)+\text{tr}(\mathbf{Y}_{\text{V},q}^{22}\mathbf{F}_{\text{P},q}),\nonumber\\
    &\text{tr}(\mathbf{Y}_{\text{V},q}\mathbf{S}_{\text{V},q})|_{\mathbf{W}_{n,q}}=\mathrm{tr}(\mathbf{B}_{n,q}\mathbf{W}_{n,q})\eta_{\text{V},n,q}.
\end{align}\end{subequations}
where the terms of $\text{tr}(\mathbf{Y}_{\text{P},q}^{11}\mathbf{X}_{\text{P},q})$ and $\text{tr}(\mathbf{Y}_{\text{V},q}^{11}\mathbf{X}_{\text{V},q})$ are not related to $\mathbf{W}_{n,q}$ and thus are excluded in $\eta_{\text{P},n,q}$ and $\eta_{\text{V},n,q}$:

\begin{subequations}\begin{align}
    \eta_{\text{P},n,q}=&\Big(\text{tr}(\mathbf{Y}_{\text{P},q}^{11}\mathbf{R}_{\text{P},n,q})+\text{tr}(\mathbf{Y}_{\text{P},q}^{12}\mathbf{R}_{\text{PV},n,q}^T)\nonumber\\
    &+\text{tr}((\mathbf{Y}_{\text{P},q}^{12})^T\mathbf{R}_{\text{PV},n,q})+\text{tr}(\mathbf{Y}_{\text{P},q}^{22}\mathbf{R}_{\text{V},n,q})\Big),\\
    \eta_{\text{V},n,q}=&\Big(\text{tr}(\mathbf{Y}_{\text{V},q}^{11}\mathbf{R}_{\text{V},n,q})+\text{tr}(\mathbf{Y}_{\text{V},q}^{12}\mathbf{R}_{\text{PV},n,q})\nonumber\\
    &+\text{tr}((\mathbf{Y}_{\text{V},q}^{12})^T\mathbf{R}_{\text{PV},n,q}^T)+\text{tr}(\mathbf{Y}_{\text{V},q}^{22}\mathbf{R}_{\text{P},n,q})\Big).
\end{align}\end{subequations}

The Lagrangian function with terms related to $\mathbf{W}_{n,q}$ can be written as:
\begin{align}
    &\negthinspace\mathcal{L}(\mathbf{W}\negthinspace_{n,q})\nonumber\\
    &=\text{tr}(\mathbf{W}\negthinspace_{n,q})\negthinspace+\negthinspace\sum_{u=1}^U\negthinspace \lambda_u\text{tr}(\mathbf{H}_{n,u}\negthinspace\mathbf{W}\negthinspace_{n,q}\negthinspace)\negthinspace-\negthinspace\mathrm{tr}(\mathbf{B}_{n,q}\negthinspace\mathbf{W}\negthinspace_{n,q})\eta_{\text{P},n,q}\negthinspace\nonumber\\
    &-\negthinspace\mathrm{tr}(\mathbf{B}_{n,q}\mathbf{W}\negthinspace_{n,q})\eta_{\text{V},n,q}\negthinspace+\negthinspace\rho_{n}\text{tr}(\mathbf{W}\negthinspace_{n,q})\negthinspace-\negthinspace\text{tr}(\mathbf{Y}\negthinspace_{n,q}\mathbf{W}\negthinspace_{n,q}).\negthinspace
\end{align}

The stationary condition is obtained by taking $\frac{\partial\mathcal{L}}{\partial\mathbf{W}_{n,q}}=\mathbf{0}$, yielding
\begin{align}
    \negthinspace\mathbf{Y}_{\negthinspace n,q}^\text{opt}&=\negthinspace( 1\negthinspace+\negthinspace\rho_{n})\mathbf{I}\negthinspace+\negthinspace\sum_{u=1}^U\negthinspace\lambda_u\mathbf{H}_{n,u}\negthinspace-\negthinspace\mathbf{B}_{n,q}(\eta_{\text{P}\negthinspace,n,q}\negthinspace+\negthinspace\eta_{\text{V}\negthinspace,n,q})\nonumber\\
    &=\mathbf{A}_{n,q}-\mathbf{B}_{n,q}(\eta_{\text{P},n,q}+\eta_{\text{V},n,q}),
\end{align}
where $\mathbf{Y}_{n,q}^\text{opt}$ denotes the optimal dual matrix corresponding to the PSD constraint of the optimal $\mathbf{W}_{n,q}^\text{opt}$.
$\mathbf{A}_{n,q}\succ 0$, $\mathrm{rank}(\mathbf{A}_{n,q})=N_\text{t}$, while $\mathrm{rank}(\mathbf{B}_{n,q})=1$ since $\mathbf{B}_{n,q}=\mathbf{b}_n(\theta_{n,q})\mathbf{b}_n^H(\theta_{n,q})$. Thus, we have $\mathrm{rank}(\mathbf{Y}_{n,q}^\text{opt})\ge N_\text{t}-1$ and $\dim(\mathcal{N}\{\mathbf{Y}_{n,q}^\text{opt}\})\le1$. 

Similar to the communication BF matrices, from the complementary slackness condition and PSD properties, we have $\text{tr}(\mathbf{Y}_{n,q}^\text{opt}\mathbf{W}_{n,q}^\text{opt})=0$ and $\mathbf{Y}_{n,q}^\text{opt}\mathbf{W}_{n,q}^\text{opt}=
\mathbf{0}$, implying $\mathrm{Col}\{\mathbf{W}_{n,q}^\text{opt}\}\subseteq \mathcal{N}\{\mathbf{Y}_{n,q}^\text{opt}\}$ and $\mathrm{rank}(\mathbf{W}_{n,q}^\text{opt})\le\dim(\mathcal{N}\{\mathbf{Y}_{n,q}^\text{opt}\})\le1$. 

In summary, it is guaranteed that the optimal solutions of (P1.3) satisfy $\mathrm{rank}(\mathbf{W}_{u}^{\text{opt}})=1$ and $\mathrm{rank}(\mathbf{W}_{n,q}^{\text{opt}})\le1$, demonstrating the tightness of the SDR.

\end{appendices}

\printbibliography

\end{document}